\documentclass[article,12pt]{amsart}
\usepackage{amsfonts}
\usepackage{amsmath}
\usepackage{graphicx}
\usepackage{color}
\usepackage{epsfig}
\usepackage{multirow}
\usepackage{chngcntr}

\numberwithin{equation}{section}

\theoremstyle{plain}
\newtheorem{thm}{Theorem}[section]
\newtheorem{prop}{Proposition}[section]

\theoremstyle{remark}
\newtheorem{defi}{Definition}[section]
\newtheorem{rem}{Remark}[section]

\newcommand{\dN}{\mathbb{N}}
\newcommand{\dZ}{\mathbb{Z}}
\newcommand{\dR}{\mathbb{R}}
\newcommand{\dC}{\mathbb{C}}
\newcommand{\dE}{\mathbb{E}}
\newcommand{\cA}{\mathcal{A}}
\newcommand{\cB}{\mathcal{B}}
\newcommand{\cC}{\mathcal{C}}
\newcommand{\cL}{\mathcal{L}}
\newcommand{\cAS}{\mathcal{A}_{s}}
\newcommand{\cBS}{\mathcal{B}_{s}}

\newcommand{\veps}{\varepsilon}

\newcommand{\wh}{\widehat}
\newcommand{\wt}{\widetilde}

\def\build#1_#2^#3{\mathrel{\mathop{\kern 0pt#1}\limits_{#2}^{#3}}}

\def\videbox{\mathbin{\vbox{\hrule\hbox{\vrule height1ex \kern.5em
\vrule height1ex}\hrule}}}

\counterwithin{figure}{section}
\counterwithin{table}{section}

\email{Sophie.Bercu@edf.fr}
\email{Frederic.Proia@inria.fr}
\keywords{SARIMA(X) modelling, Time series analysis, Exogenous covariates, Forecasting, Seasonality, Stationarity, Individual load curve}

\begin{document}

\title[On a SARIMAX coupled modelling]
{A SARIMAX coupled modelling applied to individual load curves intraday forecasting
\vspace{2ex}}
\author{Sophie Bercu}
\address{EDF Recherche \& D\'eveloppement, D\'epartement ICAME. 1, avenue du G\'en\'eral de Gaulle,
92141 Clamart Cedex.}
\author{Fr\'ed\'eric Pro\"ia}
\address{Universit\'e Bordeaux 1, Institut de Math\'ematiques de Bordeaux,
UMR 5251, and INRIA Bordeaux, team ALEA, 351 Cours de la Lib\'eration, 33405 Talence cedex, France.}
\thanks{}

\begin{abstract}
A dynamic coupled modelling is investigated to take temperature into account in the individual energy consumption forecasting. The objective is both to avoid the inherent complexity of exhaustive SARIMAX models and to take advantage of the usual linear relation between energy consumption and temperature for thermosensitive customers. We first recall some issues related to individual load curves forecasting. Then, we propose and study the properties of a dynamic coupled modelling taking temperature into account as an exogenous contribution and its application to the intraday prediction of energy consumption. Finally, these theoretical results are illustrated on a real individual load curve. The authors discuss the relevance of such an approach and anticipate that it could form a substantial alternative to the commonly used methods for energy consumption forecasting of individual customers.
\end{abstract}

\maketitle


\section{INTRODUCTION}


The electrical systems have to manage with new challenges: constant increasing demand of electricity, including the arrival of new uses such as electric vehicle, increasing production of renewable energies decentralized, increasing  number of energy market participants including aggregators, with a policy of DMS and CO$_2$ reduction. In particular, we note an increase in peaks load at critical times, intermittent injections of energy at all levels making electrical networks much more difficult to manage, the balance supply/demand more difficult to keep and economic interests much more dispersed. On the other hand, the technological advances are enabling new opportunities: the development of smart-metering and smart grid. Indeed, there is a massive deployment in Europe - 80\% of households will be equipped by 2020 - and in North America, mainly. In this context, forecasting consumption of very fine mesh becomes an important issue at the heart of the system of electric power. In this paper, we focus exclusively our attention on the prediction of end-customer on a short-term horizon. These forecasts are useful for the customers who want to optimize their bill and make DMS, for network planning and monitoring, for aggregators who wish to apply real-time load shifting. Forecasting methods on large aggregate of customers exist in abundance but cannot be extrapolated to individual curves because of their extremely irregular behavior. Indeed, aggregation has the advantage of reducing the noise and provides little chaotic curves in which trend and seasonality are easily identifiable. For example, in the case of short-term forecasting, the exponentially weighted methods of Taylor \cite{Taylor10} give pretty good results. Harvey and Koopman \cite{HarveyKoopman93} also developed an unobserved components model with time-varying splines to capture the evolution of patterns in hourly electricity loads. Afterwards, bayesian methods that rely on Kalman filter and state space models were suggested by Martin \cite{Martin99} and Smith \cite{Smith00}. Multiple time series approaches have also emerged to model and predict intraday aggregated load curves, one can cite in example the heteroskedastic GARCH model of Garcia \textit{et al.} \cite{GarciaContrerasVanAkkerenGarcia05}, the seasonal ARFIMA-GARCH model of Koopman \textit{et al.} \cite{KoopmanOomsCarnero07}, the robust SARIMA estimates of Chakhchoukh \textit{et al.} \cite{ChakhchoukhPanciaticiBondon09}, \cite{ChakhchoukhPanciaticiMili09}, or the bias reducing iterative algorithm of Cornillon \textit{et al.} \cite{CornillonHengartnerLefieuxMatzner09}. Lately, Dordonnat \textit{et al.} \cite{DordonnatKoopmanOomsDessertaineCollet08} have developed a very comprehensive space state approach for modelling hourly french national electricity load, taking into account different levels of seasonality, calendar events and weather dependence, with one equation for each hour. Semiparametric methods and artificial neural networks to model meteorological effects and seasonal patterns are considered by Cottet and Smith \cite{CottetSmith03}, and by Liu \textit{et al.} \cite{LiuChenLiuHarris06}. Poggi \cite{Poggi94} proposed a nonparametric approach based on kernel estimators to make forecasts on half-hourly french load curves and more technical studies from Antoniadis \textit{et al.} \cite{AntoniadisPaparoditisSapatinas06}, \cite{AntoniadisPaparoditisSapatinas09} based on functional kernel-wavelets can also be mentioned. In short, several methods have been implemented of various kinds, from nonparametric to parametric models with exogenous covariates through semi-parametric approaches, heteroskedasticity, space state models and neural networks, each of them providing excellent intraday forecast results on extensively averaged curves, such as national load. Let us conclude by refering the reader to the approach of Devaine \textit{et al.} \cite{DevaineGoudeStoltz09}, which is an aggregation of specialized experts combining a set of prediction outputs by independent forecasters.

\medskip

The issue of individual forecasting is complex and very few literature is available on the subject. The main difficulty of the individual load curves is the deep irregularity resulting from the human behavior. Indeed, we have to deal with phenomena that aggregation usually hides, such as high disturbances, unpredictable local behaviors or thresholds during holiday periods. To the best of our knowledge, very few studies have been conducted in this area. In their work, Espinoza \textit{et al.} \cite{EspinozaJoyeBelmansDeMoor05} are concerned by the short-term load forecasting from a HV-LV substation, and Ghofrani \textit{et al.} \cite{GhofraniHassanzadehEtezadiFadali11} propose to model real-time measurement data from customers' smart meters as the sum of a deterministic component and a gaussian noise signal. This paper suggests a statistical parametric approach adapted to an individual load curve which shows a substantial seasonal pattern and a thermosensitive behavior as prerequisites, highly relying on the time series theory. The authors anticipate that it could form an alternative to the commonly used methods for energy consumption forecasting of individual customers. In Section 2, we introduce a dynamic coupled modelling taking temperature into account. We also recall some known results on time series analysis, in particular the concepts of stationarity and causality. We recall some theoretical backgrounds which will be used as a basis for the empirical study on a real curve in the next section. As an exemple of load curve which we will investigate, Figure \ref{Fig_Ex} displays the energy consumption of a thermosensitive customer and the interpolated temperatures measured every 3 hours by the nearest weather station on the same period of 6 months. Section 3 is devoted to the detailed study on such a load curve based on time series analysis. The main motivations for proposing a coupled modelling are the linear relation between the logarithmic energy consumption and the temperature on the one hand, and the seasonal behavior of the residuals on the other hand. Figure \ref{Fig_ConsoTemp} represents the scatter plot between temperature and consumption for the same customer, in which one can observe the linear relationship. It also displays the residuals from the linear regression on the right-hand side. One shall investigate seasonality, stationarity and autocorrelations in the residuals from the linear regression, to build a suitable time series modelling and propose a forecasting algorithm, according to some criteria that will be specified. For that purpose, one shall make an extensive use of the well-known Box and Jenkins methodology \cite{BoxJenkinsReinsel76}. A short conclusion is given in Section 4.
\begin{figure}[h!]
\includegraphics[width=7.2cm]{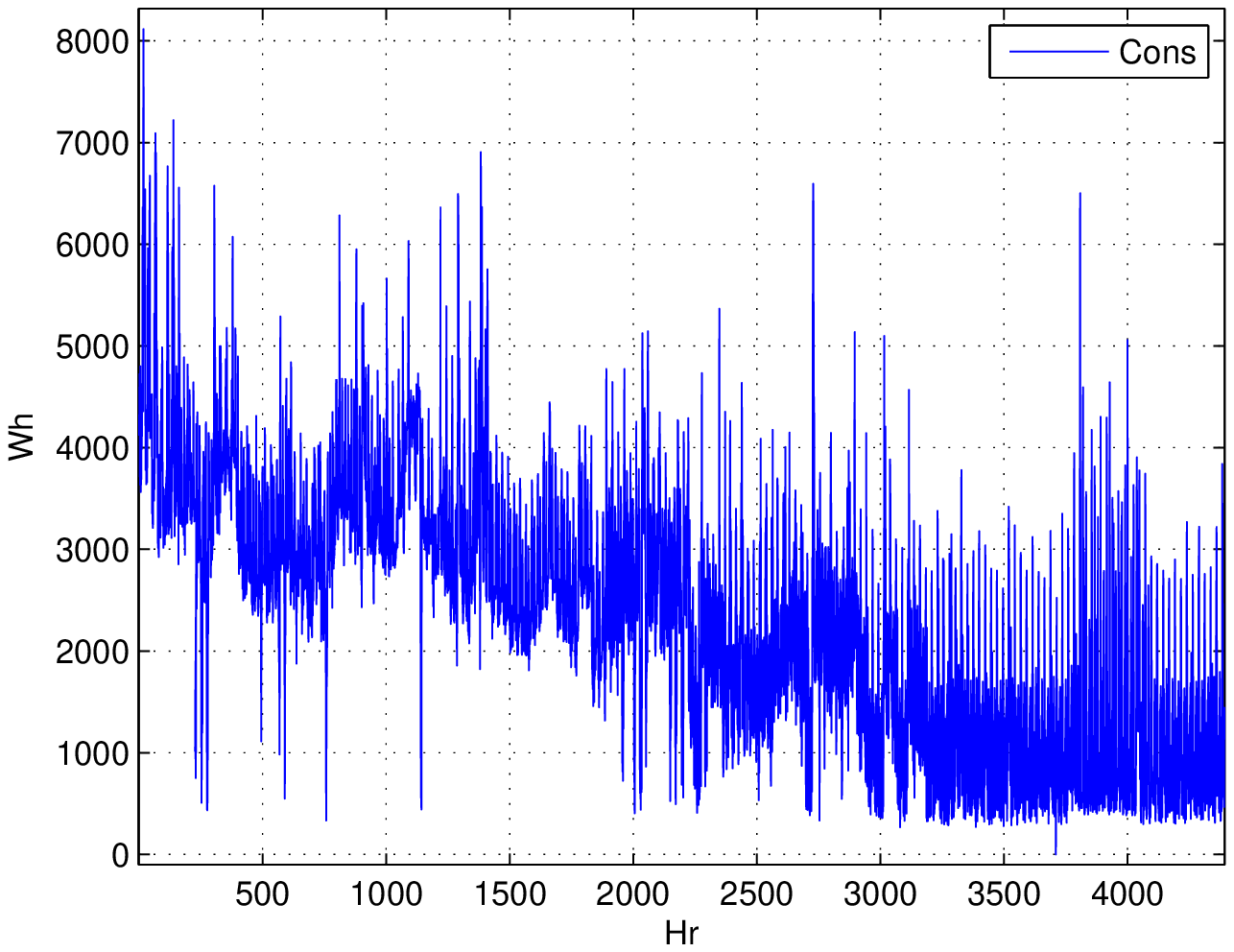} \includegraphics[width=7.2cm]{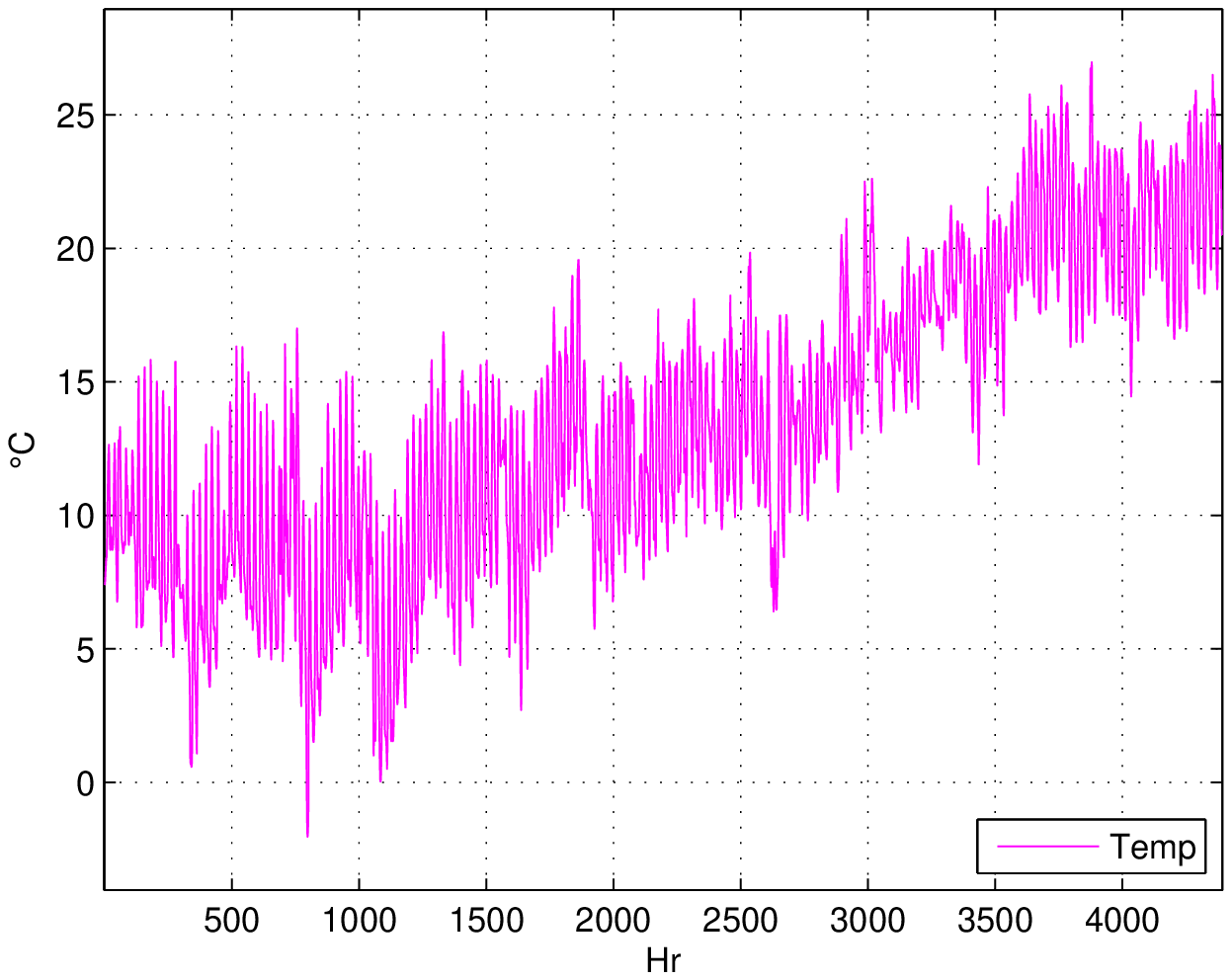}
\vspace{-0.5cm} \\
\caption{Energy consumption of a thermosensitive customer (left), temperature during the same period (right).}
\label{Fig_Ex}
\end{figure}
\begin{figure}[h!]
\includegraphics[width=7.2cm]{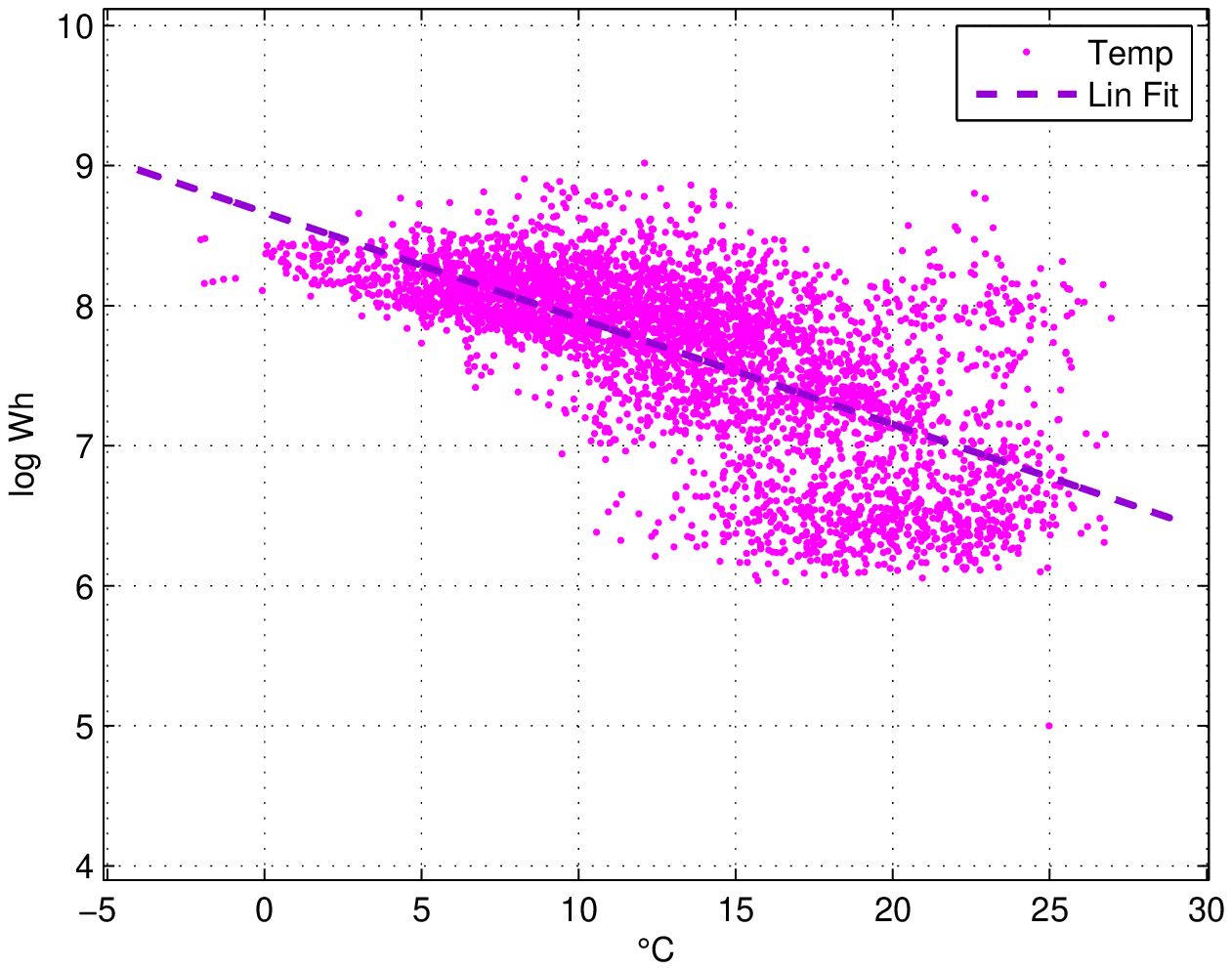} \includegraphics[width=7.2cm]{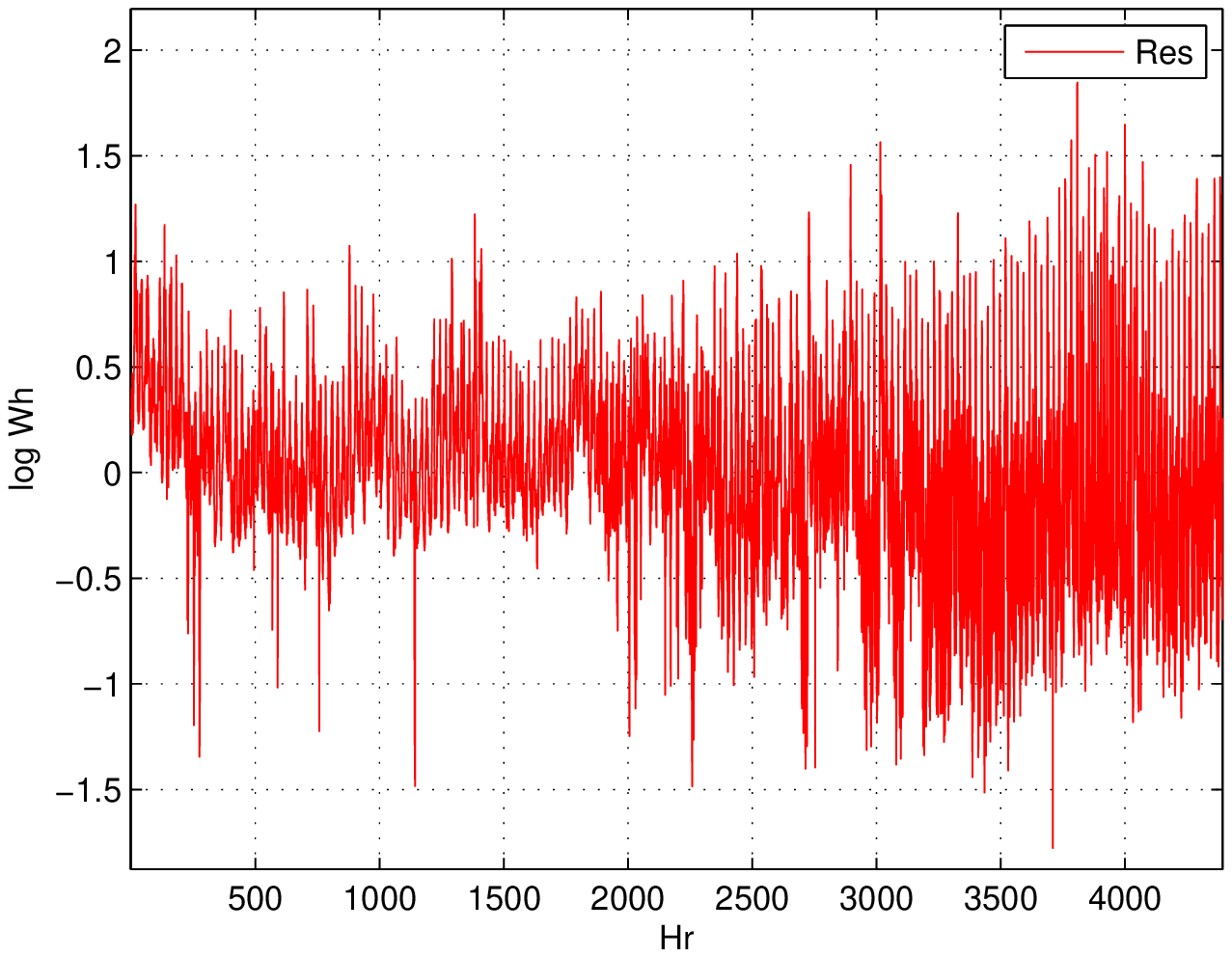}
\vspace{-0.5cm} \\
\caption{Scatter plot between energy consumption and temperature (left), residuals from the linear regression (right).}
\label{Fig_ConsoTemp}
\end{figure}

\begin{rem}
In all the sequel, $B$ stands for the backshift operator which operates on an element of a given time series to produce the previous element, $B X_{t} = X_{t-1}$. The backshift operator is raised to arbitrary integer powers so that $B^{s} X_{t} = X_{t-s}$. The difference operator $\nabla$, defined as $\nabla X_{t} = (1 - B) X_{t}$, also generalises to arbitrary integer powers so that $\nabla_{\! s} X_{t} = (1 - B^{s}) X_{t}$ and $\nabla^{s} X_{t} = (1 - B)^{s} X_{t}$.
\end{rem}

\bigskip


\section{ON A SARIMAX COUPLED MODELLING}


Let us start by recalling some usual tools related to time series analysis that we shall make repeatedly use of throughout the study. The reader will find more details on these results in Chapters 1 and 3 of \cite{BrockwellDavis91}.
\begin{defi}[Stationarity]
A time series $(Y_{t})$ is said to be weakly stationary if, for all $t \in \dZ$,  $\dE[Y_{t}^2] < \infty$, $\dE[Y_{t}] = m$ and, for all $s,t \in \dZ$, $\dE[Y_{t} Y_{s}] = \dE[Y_{t-s}  Y_{0}]$.
\end{defi}
\noindent In all the sequel, the term \textit{stationarity} will always refer to the weak stationarity. Let us now focus on the stationary ARMA process and on the concept of causality.
\begin{defi}[ARMA]
Let $(Y_{t})$ be a stationary time series with zero mean. It is said to be an ARMA$(p,q)$ process if, for every $t \in \dZ$,
\begin{equation}
\label{Arma}
Y_{t} - \sum_{k=1}^{p} a_{k} Y_{t-k} = \veps_{t} + \sum_{k=1}^{q} b_{k} \veps_{t-k}
\end{equation}
where $(\veps_{t})$ is a white noise of variance $\sigma^2 > 0$ and the parameters $a \in \dR^{p}$ and $b \in \dR^{q}$.
\end{defi}
\noindent The equation \eqref{Arma} can be rewritten in the more compact form
\begin{equation*}
\cA(B) Y_{t} = \cB(B) \veps_{t}
\end{equation*}
where the polynomials
\begin{equation*}
\cA(z) = 1 - a_{1} z - \hdots - a_{p} z^{p} \hspace{0.5cm} \text{and} \hspace{0.5cm} \cB(z) = 1 + b_{1} z + \hdots + b_{q} z^{q}.
\end{equation*}
In the particular case where $p=0$, $\cA(z) = 1$ and $(Y_{t})$ is a moving average MA$(q)$ process. Likewise, if $q=0$, $\cB(z) = 1$ and $(Y_{t})$ is an autoregressive AR$(p)$ process.
\begin{defi}[Causality]
Let $(Y_{t})$ be an ARMA$(p,q)$ process for which the polynomials $\cA$ and $\cB$ have no common zeroes. Then, $(Y_{t})$ is causal if and only if $\cA(z) \neq 0$ for all $z \in \dC$ such that $\vert z \vert \leq 1$.
\end{defi}
\noindent The causality enables to write the ARMA$(p,q)$ process as an MA$(\infty$) one. It guarantees the existence of an unique stationary solution for the ARMA$(p,q)$ process expressed as a linear process associated with $(\veps_{t})$, by virtue of the following result.
\begin{prop}
\label{Prop_StatARMA}
If $\cA(z) \neq 0$ for all $z \in \dC$ such that $\vert z \vert \leq 1$, then the ARMA equation $\cA(B) Y_{t} = \cB(B) \veps_{t}$ has the unique stationary solution
\begin{equation}
\label{MA_inf}
Y_{t} = \sum_{k=0}^{\infty} \psi_{k} \veps_{t-k},
\end{equation}
and the coefficients $(\psi_{k})_{k \in \dN}$ are determined by the relation
\begin{equation*}
\cA^{-1}(z) \cB(z) = \sum_{k=0}^{\infty} \psi_{k} z^{k} \hspace{1cm} \text{with} \hspace{1cm} \sum_{k=0}^{\infty} \vert \, \psi_{k} \vert < \infty.
\end{equation*}
\end{prop}
\begin{proof}
We refer the reader to Theorem 3.1.1 of \cite{BrockwellDavis91} for a detailed version of Proposition \ref{Prop_StatARMA}, and the proof that follows.
\end{proof}
\noindent One can observe that the stationarity of the solution and the causality of the ARMA process will often coincide on the real curves defined on the positive integers that will be considered in the following. Indeed, a zero inside the unit circle results in an explosive behavior of the process that cannot match with stationarity on $\dN^{\,*}$. We now focus our attention on some practical tools related to identification methods for the orders of stationary AR$(p)$ and MA$(q)$ processes.
\begin{defi}[ACF] Let $(Y_{t})$ be a stationary time series. The autocorrelation function $\rho$ associated with $(Y_{t})$ is defined, for all $t \in \dZ$, as
\begin{equation*}
\rho(t) = \frac{\gamma(t)}{\gamma(0)}
\end{equation*}
where the autocovariance function $\gamma(t) = \dE[Y_{t} Y_0] - \dE[Y_{t}] \dE[Y_0]$.
\end{defi}
\begin{prop}
\label{Prop_MA}
The stationary time series $(Y_{t})$ with zero mean is a MA$(q)$ process such that $b_{q} \neq 0$ if and only if $\rho(q) \neq 0$ and $\rho(t) = 0$ for all $\vert t \vert > q$.
\end{prop}
\noindent Denote by $(\phi_{t,\, j})$, for all $t \in \dN^{\,*}$ and $j \in \{ 1, \hdots, t \}$, the sequence given by $\phi_{1,1} = \rho(1)$ and, for all $t \geq 2$, by the Levinson-Durbin recursion,
\begin{equation*}
\phi_{t,\, t} = \left( 1 - \sum_{k=1}^{t-1} \phi_{t-1,\, k}\, \rho(k) \right)^{\hspace{-0.15cm} -1} \left( \rho(t) - \sum_{k=1}^{t-1} \phi_{t-1,\, k}\, \rho(t-k) \right)
\end{equation*}
where, for all $k \in \{ 1, \hdots, t-1 \}$, $\phi_{t,\, k} = \phi_{t-1,\, k} - \phi_{t,\, t} \phi_{t-1,\, t-k}$. The sequence $(\phi_{t,\, t})_{t \in \dN^{\,*}}$ may be seen as the correlation between two residuals obtained after regressing $Y_{t+1}$ and $Y_1$ on the intermediate observations $Y_2, \hdots, Y_{t}$. Formally, for all $t \in \dN^{\,*}$, 
\begin{equation*}
\phi_{t,\, t} = \text{Corr}(Y_{t+1} - \text{P}_{\bar{\text{sp}} \{Y_2,\, \hdots,\, Y_{t}\}} Y_{t+1}, Y_1 - \text{P}_{\bar{\text{sp}} \{Y_2,\, \hdots,\, Y_{t}\}} Y_1)
\end{equation*}
where $\text{P}_{\bar{\text{sp}} \{Y_2,\, \hdots,\, Y_{t}\}}$ stands for the $L^2-$orthogonal projection operator of any square-integrable random variable on the closed subspace generated by $Y_2,\, \hdots,\, Y_{t}$.
\begin{defi}[PACF] Let $(Y_{t})$ be a stationary time series with zero mean. The partial autocorrelation function $\alpha$ is defined as $\alpha(0) = 1$, and, for all $t \in \dN^{\,*}$, as
\begin{equation*}
\alpha(t) = \phi_{t,\, t}.
\end{equation*}
\end{defi}
\begin{prop}
\label{Prop_AR}
If there exists a sequence $(\psi_{k}) \in \ell^2(\dN)$ such that $(Y_{t})$ has the unique expression given by \eqref{MA_inf} with $\psi_0=1$, then the stationary time series $(Y_{t})$ with zero mean is an AR$(p)$ process such that $a_{p} \neq 0$ if and only if $\alpha(p) \neq 0$ and $\alpha(t) = 0$ for all $t > p$.
\end{prop}
\begin{proof}
The proofs of Propositions \ref{Prop_MA} and \ref{Prop_AR} may be found in Chapter 3 of \cite{BrockwellDavis91}.
\end{proof}
\noindent These identification techniques will be useful thereafter for orders selection in models building. The hypothesis of stationarity will be tested \textit{via} the commonly used Kwiatkowski-Phillips-Schmidt-Shin \textit{KPSS test} \cite{KwiatkowskiPhillipsSchmidtShin92} together with the unit root Augmented Dickey-Fuller \textit{ADF test} \cite{DickeySaid84}. As for the hypothesis of white noise, it will be evaluated through the \textit{portmanteau test} of Ljung-Box \cite{LjungBox78}, \cite{BoxPierce70}.

\medskip

In all the sequel, we denote by $(C_{t})$ the individual energy consumption of a given customer, for all $1 \leq t \leq T$. We also denote by $(U_{t})$ the temperature associated with $(C_{t})$, supposed to be known for all $-r+2 \leq t \leq T+H$ where $H$ is the prediction horizon. In addition, we will use a variance-stabilizing Box-Cox logarithmic transformation $(Y_{t})$ ensuring homoskedasticity, given, for all $1 \leq t \leq T$, by
\begin{equation}
\label{Y}
Y_{t} = \log \left( C_{t} + \mathrm{e}^{\mu} \right)
\end{equation}
where $\mu$ is a positive parameter to evaluate, implying that $Y_{t}=\mu$ in the particular case where $C_{t}=0$. This safety precaution is justified by the possible use of relative criteria, such as the \textit{Mean Absolute Percentage Error}.

\medskip

\noindent{\bf The dynamic coupled modelling.} The first step of the modelling relies in a suitable way to remove the direct influence of the temperature on the consumption. As mentioned above, it exists a strong correlation between $(Y_{t})$ and $(U_{t})$. This relationship is modeled through the linear regression given, for all $1 \leq t \leq T$, by
\begin{equation}
\label{ModLR}
Y_{t} = c_0 + \cC(B) U_{t} + \veps_{t}
\end{equation}
where $c_0 \in \dR$ is an intercept, $\cC(B)$ is a polynomial of order $r$ such that, for all $z \in \dC$,
\begin{equation*}
\cC(z) = \sum_{k=1}^{r} c_{k} z^{k-1}
\end{equation*}
and the unknown vector parameter $c \in \dR^{\, r+1}$ is estimated by standard least squares. The disturbance terms $(\veps_{t})$ will be regarded as a seasonal time series. In particular, $(\veps_{t})$ is said to follow a SARIMA$(p,d,q) \times (P,D,Q)_{s}$ modelling if, for all $1 \leq t \leq T$,
\begin{equation}
\label{ModTS}
(1 - B)^{d} (1 - B^{s})^{D} \cA(B) \cAS(B) \veps_{t} = \cB(B) \cBS(B) V_{t},
\end{equation}
according to Definition 9.6.1 of \cite{BrockwellDavis91}, where $(V_{t})$ is a white noise of variance $\sigma^2 > 0$, and where the polynomials are defined, for all $z \in \dC$, as
\begin{equation*}
\cA(z) = 1 - \sum_{k=1}^{p} a_{k} z^{k}, \hspace{1cm} \cAS(z) = 1 - \sum_{k=1}^{P} \alpha_{k} z^{s k},
\end{equation*}
\begin{equation*}
\cB(z) = 1 - \sum_{k=1}^{q} b_{k} z^{k}, \hspace{1cm} \cBS(z) = 1 - \sum_{k=1}^{Q} \beta_{k} z^{s k}.
\end{equation*}
In this modelling, $a \in \dR^{p}$, $b \in \dR^{q}$, $\alpha \in \dR^{P}$ and $\beta \in \dR^{Q}$ are vector parameters estimated by generalized least squares.
\noindent The differenced process $(\nabla^{d} \nabla_{\! s}^{D} \veps_{t})$ in \eqref{ModTS} is a stationary solution of the ARMA causal process, \textit{i.e.} $\cA(z) \neq 0$ and $\cAS(z) \neq 0$ for all $z \in \dC$ such that $\vert z \vert \leq 1$.
\begin{defi}[SARIMAX]
In the particular framework of the study, a random process $(Y_{t})$ will be said to follow a SARIMAX$(p,d,q,r) \times (P,D,Q)_{s}$ coupled modelling if, for all $1 \leq t \leq T$, it satisfies
\begin{equation}
\label{ModCou}
\vspace{1ex}
\left\{
\begin{array}[c]{l}
Y_{t} = c_0 + \cC(B) U_{t} + \veps_{t}, \vspace{0.2cm} \\
(1 - B)^{d} (1 - B^{s})^{D} \cA(B) \cAS(B) \veps_{t} = \cB(B) \cBS(B) V_{t}.
\end{array}
\right.
\end{equation}
\end{defi}
\noindent The orders $p$, $d$, $q$, $r$, $P$, $D$, $Q$ and $s$ shall be evaluated following a well-known Box and Jenkins methodology \cite{BoxJenkinsReinsel76}. Moreover, a straightforward calculation shows that \eqref{ModCou} can be rewritten in the condensed form given, for all $1 \leq t \leq T$, by
\begin{equation}
\label{ModGen}
(1 - B)^{d} (1 - B^{s})^{D} \cA(B) \cAS(B) \left( Y_{t} - \cC(B) U_{t} \right)  = \cB(B) \cBS(B) V_{t},
\end{equation}
as soon as $d+D > 0$, which will be an assumption always verified as we shall explain in the next section. Indeed, $c_0$ vanishes by a single differentiation of $(\veps_{t})$. In light of foregoing, one can establish the following result, denoting by $I$ the identity matrix of order $T$, $Y$ and $U$ the observation vector of order $T$ and the design matrix of order $T \times (r+1)$, respectively given by
\begin{equation*}
Y = \begin{pmatrix}
Y_1 \\
Y_2 \\
\vdots \\
Y_{T}
\end{pmatrix} \hspace{0.5cm} \text{and} \hspace{0.5cm}
U = \begin{pmatrix}
1 & U_{T} & U_{T-1} & \hdots & U_{T-r+1} \\
1 & U_{T-1} & U_{T-2} & \hdots & U_{T-r} \\
\vdots & \vdots & \vdots & & \vdots \\
1 & U_{1} & U_{0} & \hdots & U_{-r+2} 
\end{pmatrix}.
\end{equation*}
\begin{thm}
\label{Thm_ExistSol}
Assume that $U^{\, \prime} U$ is invertible. Then, the differenced process $(\nabla^{d} \nabla_{\! s}^{D} \veps_{t})$ where $\veps_{t}$ is given, for all $1 \leq t \leq T$, by the vector form
\begin{equation}
\label{ResTh}
\veps = \left( I - U (U^{\, \prime} U)^{-1} U^{\, \prime} \right) Y
\end{equation}
is a stationary solution of the coupled model \eqref{ModCou}.
\end{thm}
\begin{proof}
Theorem \ref{Thm_ExistSol} is a direct consequence of Proposition \ref{Prop_StatARMA} together with a straightforward least squares calculation.
\end{proof}
\begin{rem}
In the particular case where $r=0$, we merely obtain $\veps = Y - \bar{Y}$ where
\begin{equation*}
\bar{Y} = \frac{1}{T} \sum_{k=1}^{T} Y_{k},
\end{equation*}
and \eqref{ModCou} reduces to the usual SARIMA$(p,d,q) \times (P,D,Q)_{s}$ modelling on the recentered load curve. In addition, as soon as $d+D>0$, the influence of $\bar{Y}$ vanishes.
\end{rem}

\noindent The $t-$statistic associated with each parameter, exploiting the asymptotic normality of the estimates, will provide a significance testing procedure, as a confirmation of the criteria minimization strategy. Though, they will not be appropriate in the exogenous regression owing to the strong autocorrelation in the residuals, and will only be applied to the time series coefficients.

\medskip

\noindent{\bf Application to forecasting.} Whatever prediction method one wishes to apply, see \textit{e.g.} Chapters 5 and 9 of \cite{BoxJenkinsReinsel76} or Chapter 5 of \cite{BrockwellDavis91}, the time series analysis of \eqref{ModCou} provides the predictor of $(\veps_{t})$ at stage $T+1$, denoted by $\wt{\veps}_{T+1}$. Let $\wh{c}_{T}$ be the least squares estimator of $c$ in \eqref{ModLR} and assume that the order $r$ is known. Then, it follows that
\begin{equation}
\label{Pred_H1}
\wt{Y}_{T+1} = \wh{c}_{0,\,T} + \sum_{k=1}^{r} \wh{c}_{k,\,T} U_{T-k+2} + \wt{\veps}_{T+1}.
\end{equation}
\textit{Via} the same lines, since $(U_{t})$ is supposed to be known for all $-r+2 \leq t \leq T+H$, the predictor at horizon $H$ is given by
\begin{equation}
\label{Pred_HH}
\wt{Y}_{T+H} = \wh{c}_{0,\,T} + \sum_{k=1}^{r} \wh{c}_{k,\,T} U_{T-k+H+1} + \wt{\veps}_{T+H}.
\end{equation}

\bigskip


\section{APPLICATION TO FORECASTING ON A LOAD CURVE}


By virtue of Theorem \ref{Thm_ExistSol}, the application of the coupled model \eqref{ModCou} to real curves merely consists in identifying the seasonality and the orders of differencing ensuring the stationarity of the residual sequence from the regression analysis. Moreover, from a careful analysis of the ACF and PACF, we will get a first approximation of the orders to be considered in the ARMA modelling. We shall first investigate seasonality through a Fourier spectrogram, then stationarity of the deseasonalized series and autocorrelations \textit{via} ACF and PACF, and finally the overall randomness of successive innovations. Different models will be suggested and compared using bayesian criteria on the one hand, and then prediction criteria on the other hand. As mentioned above, the KPSS test \cite{KwiatkowskiPhillipsSchmidtShin92}, the ADF test \cite{DickeySaid84} and the Ljung-Box test \cite{LjungBox78}, \cite{BoxPierce70} will be used as statistical procedures for evaluating the hypothesis of stationarity, of unit root and of white noise up to a certain lag, respectively.

\medskip

From now on, for all $1 \leq t \leq T$, $(C_{t})$ is a load curve and $(Y_{t})$ is the associated logarithmic process, given by \eqref{Y} with $\mu=5$. In addition, $(U_{t})$ is the exogenous temperature supposed to be known for all $-r+2 \leq t \leq T+H$ and $H$ is the prediction horizon. Denote also by $(\wh{\veps}_{t})$ the least squares estimated residual set from the regression analysis accordingly given, for all $1 \leq t \leq T$, by
\begin{equation}
\label{ResEst}
\wh{\veps}_{t} = Y_{t} - \wh{c}_{0,\,T} - \sum_{k=1}^{r} \wh{c}_{k,\,T} U_{t-k+1},
\end{equation}
directly coming from \eqref{ResTh}. For a sake of simplicity, one shall take $r=1$ without loss of generality. Besides, one observes on real curves that numerical results are very similar when $r$ increases. Indeed, due to the natural phenomenon it represents, temperature $U_{t}$ at time $t$ is highly correlated to $U_{t-1}$ and the use of lots of regressors to explain $C_{t}$ in our modelling would often be redundant and generate statistically nonsignificant coefficients.

\medskip

\noindent{\bf Seasonality.} Let us choose for example $T = 730$, that is 2 years of consumption. We consider the $k-$th empirical Fourier coefficient of $(\wh{\veps}_{t})$, given by
\begin{equation*}
\gamma_{k} = \frac{1}{2 \pi T} \left\vert \sum_{t=1}^{T} \wh{\veps}_{t} \, \mathrm{e}^{- \mathrm{i} f_{k} t} \right\vert^2
\end{equation*}
where $f_{k} = 2 \pi k / T$ is the $k-$th Fourier frequency. Figure \ref{Fig_FourRes} displays the variation of $\sqrt{\gamma_{k}}$ on the Fourier frequency spectrum of $(\wh{\veps}_{t})$ on the left-hand side and the ones of $(\nabla_{\! 12} \wh{\veps}_{t})$ and $(\nabla_{\! 24} \wh{\veps}_{t})$ on the right-hand side, with unexploitable low frequencies truncated.
\begin{figure}[h!]
\includegraphics[width=7.2cm]{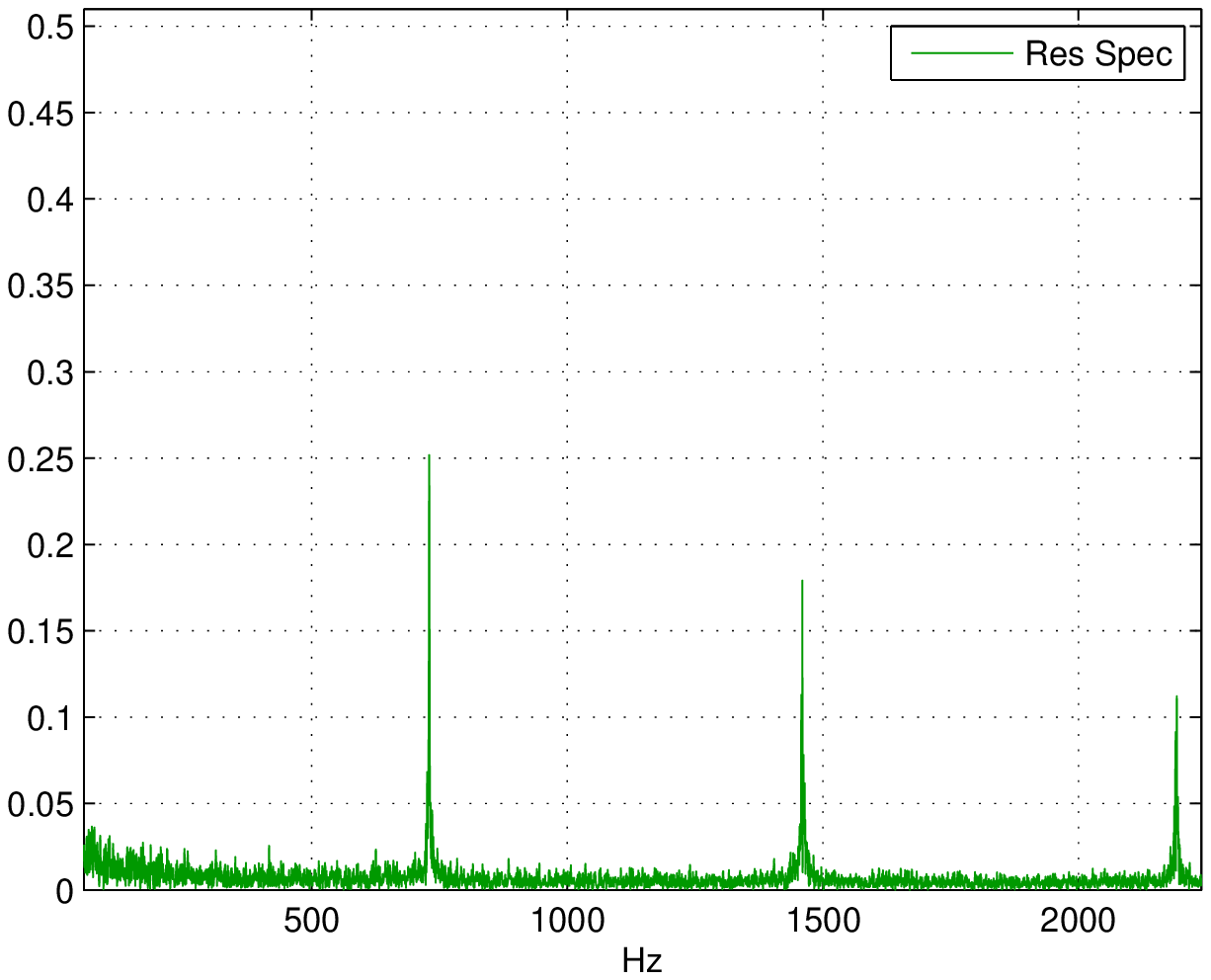} \includegraphics[width=7.2cm]{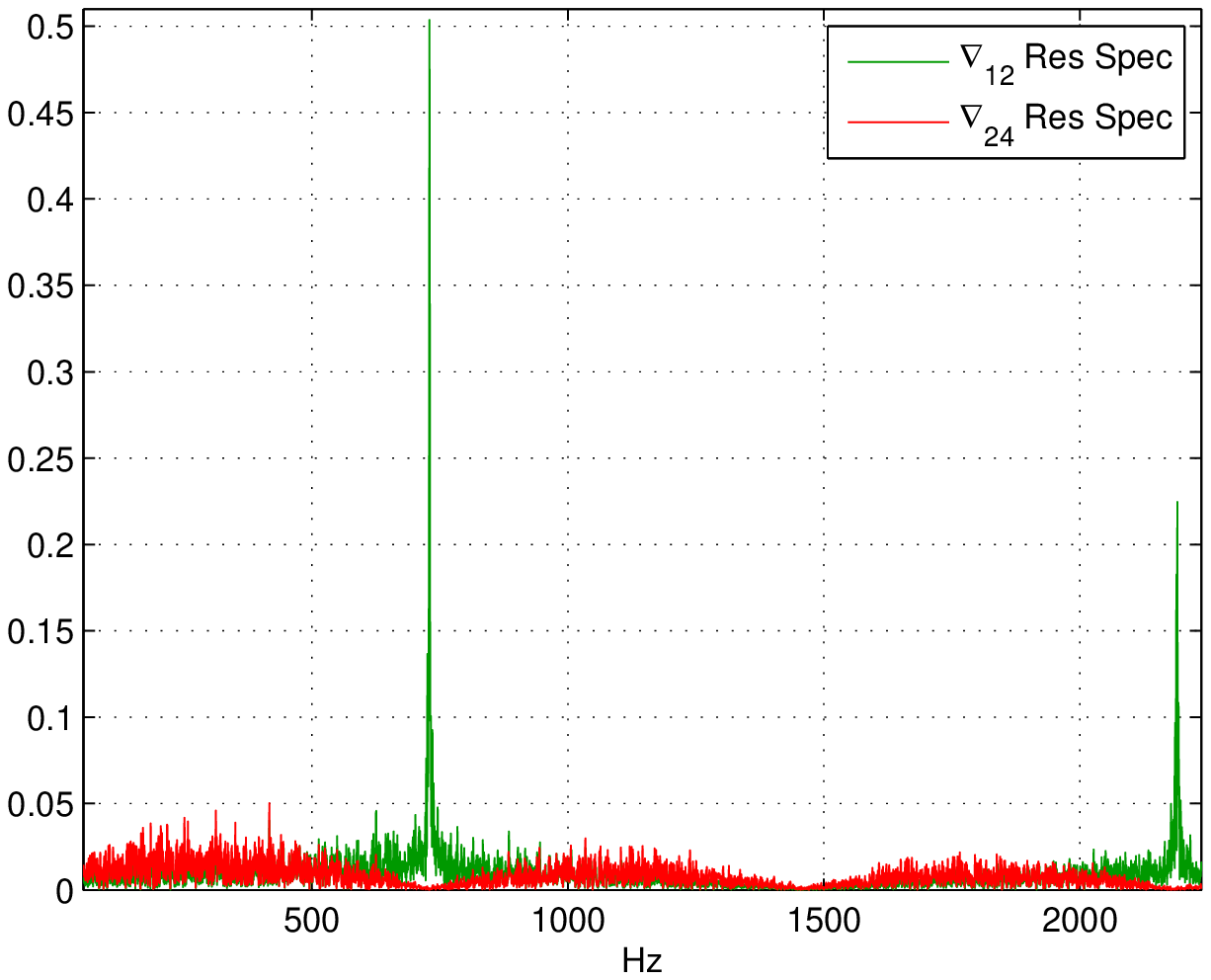}
\caption{Fourier spectrograms of residuals (left) and seasonally differenced residuals of period 12 and 24 (right).}
\label{Fig_FourRes}
\end{figure}

\medskip

\noindent Figure \ref{Fig_FourRes} shows that the estimated residual set $(\wh{\veps}_{t})$ has a seasonality and the abscissa of the main peak indicates that the pattern repeats itself 730 times on 2 years, that is daily. The second peak also suggests a seasonality of 12 hours. On the right side, one can see that $(\nabla_{\! 12} \wh{\veps}_{t})$ still has a periodicity whereas $(\nabla_{\! 24} \wh{\veps}_{t})$ is quasi-aperiodic. This is the reason why one shall choose $s=24$ in the SARIMA modelling, and that also leads us to choose $D=1$ in \eqref{ModCou}.

\medskip

\noindent{\bf Stationarity.} The KPSS and ADF statistical procedures both suggest that, on 6 months of consumption, $(\wh{\veps}_{t})$ is not stationary whereas $(\nabla \wh{\veps}_{t})$, $(\nabla_{\! 24} \wh{\veps}_{t})$ and $(\nabla \nabla_{\! 24} \wh{\veps}_{t})$ are stationary. As a consequence, $(\wh{\veps}_{t})$ is difference-stationary and the differenced series can all be solutions of a causal ARMA modelling, which leads to ARIMA models with $d=1$ and SARIMA models with $d=0$ and $D=1$ or $d=1$ and $D=1$.

\medskip

\noindent{\bf Autocorrelations.} On the ACF and PACF of $(\wh{\veps}_{t})$, one can clearly observe the daily periodicity of the series, as it appears on Figure \ref{Fig_ResACF}.
\begin{figure}[h!]
\includegraphics[width=7.2cm]{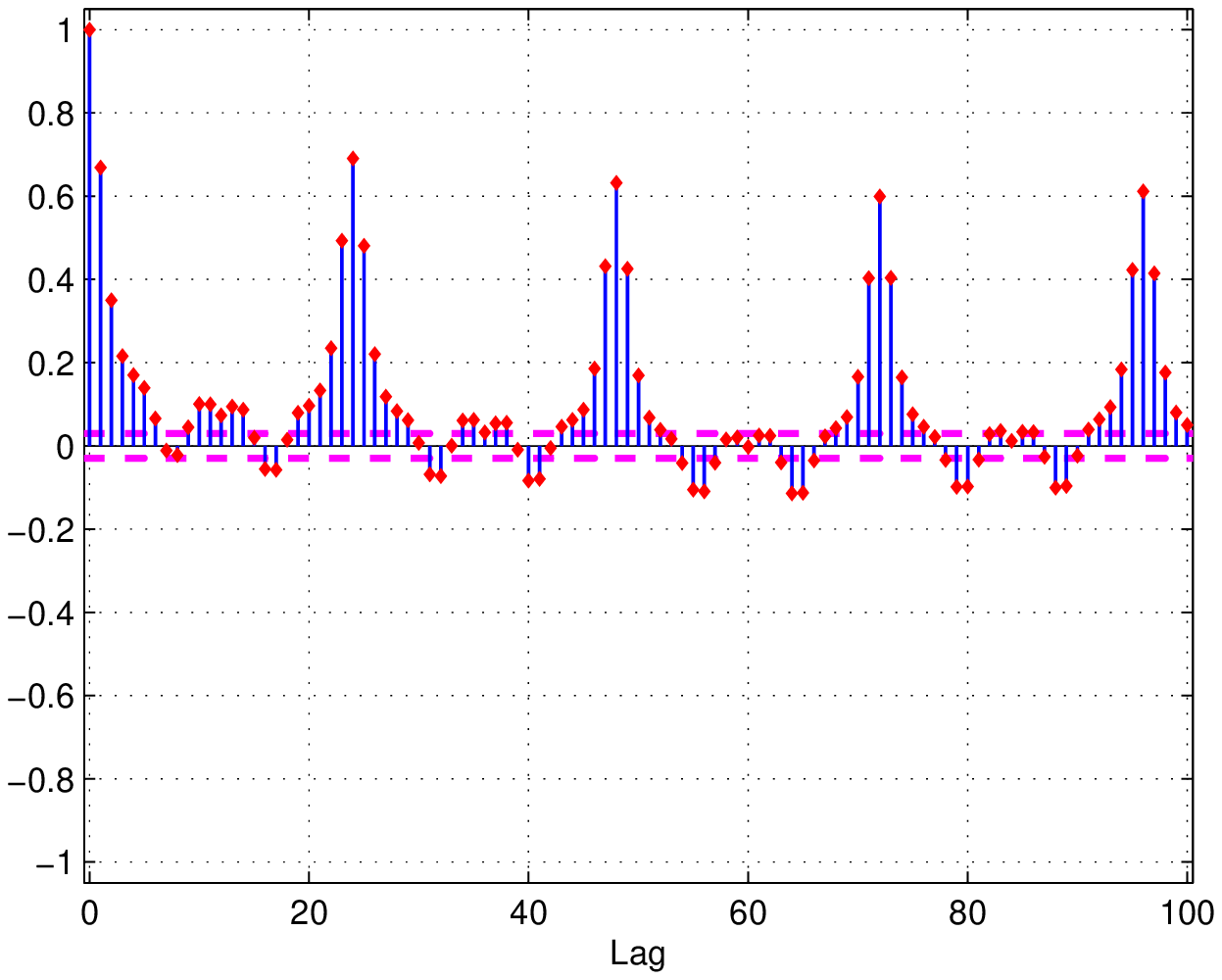} \includegraphics[width=7.2cm]{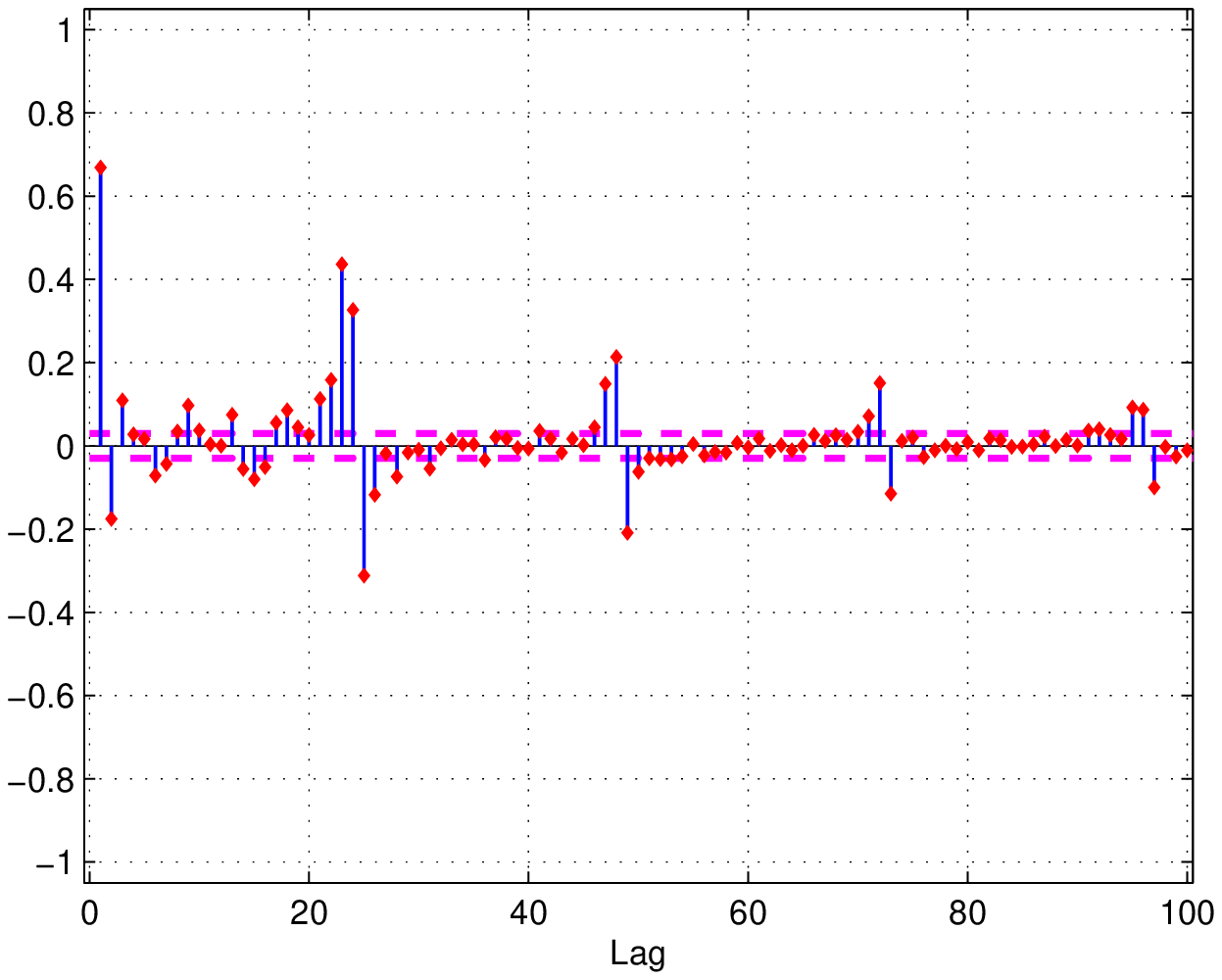}
\caption{ACF (left) and PACF (right) of the residuals $(\wh{\veps}_{t})$.}
\label{Fig_ResACF}
\end{figure}

\newpage

\noindent In addition, the sample ACF of $(\nabla_{\! 24} \wh{\veps}_{t})$ on Figure \ref{Fig_ResDiffsACF} below shows either an exponential decay or a mixture with damped sine wave, while the sample PACF has a relatively large spike at lag 1 and can reasonably be considered as nonsignificant afterwards, with uncertainty up to lag 5. One can also detect a pattern around lag 24 on the ACF. This behavior suggests an AR$(p)$ modelling with a seasonal moving average autoregression on the seasonally differenced series, that is a SARIMA$(p,0,0) \times (0,1,Q)_{24}$ model with $p \leq 5$, and $Q = 1$ on $(\wh{\veps}_{t})$.
\begin{figure}[h!]
\includegraphics[width=7.2cm]{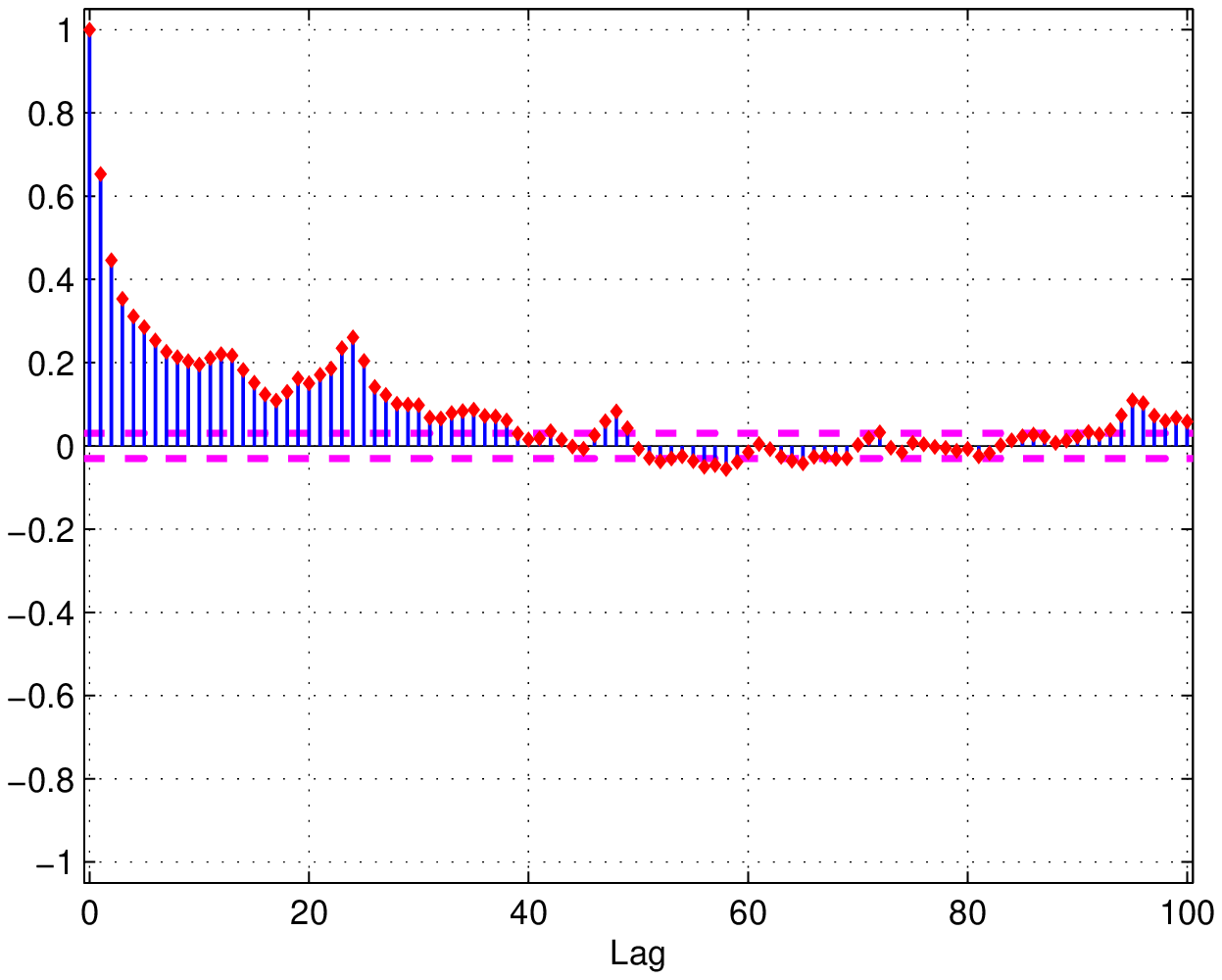} \includegraphics[width=7.2cm]{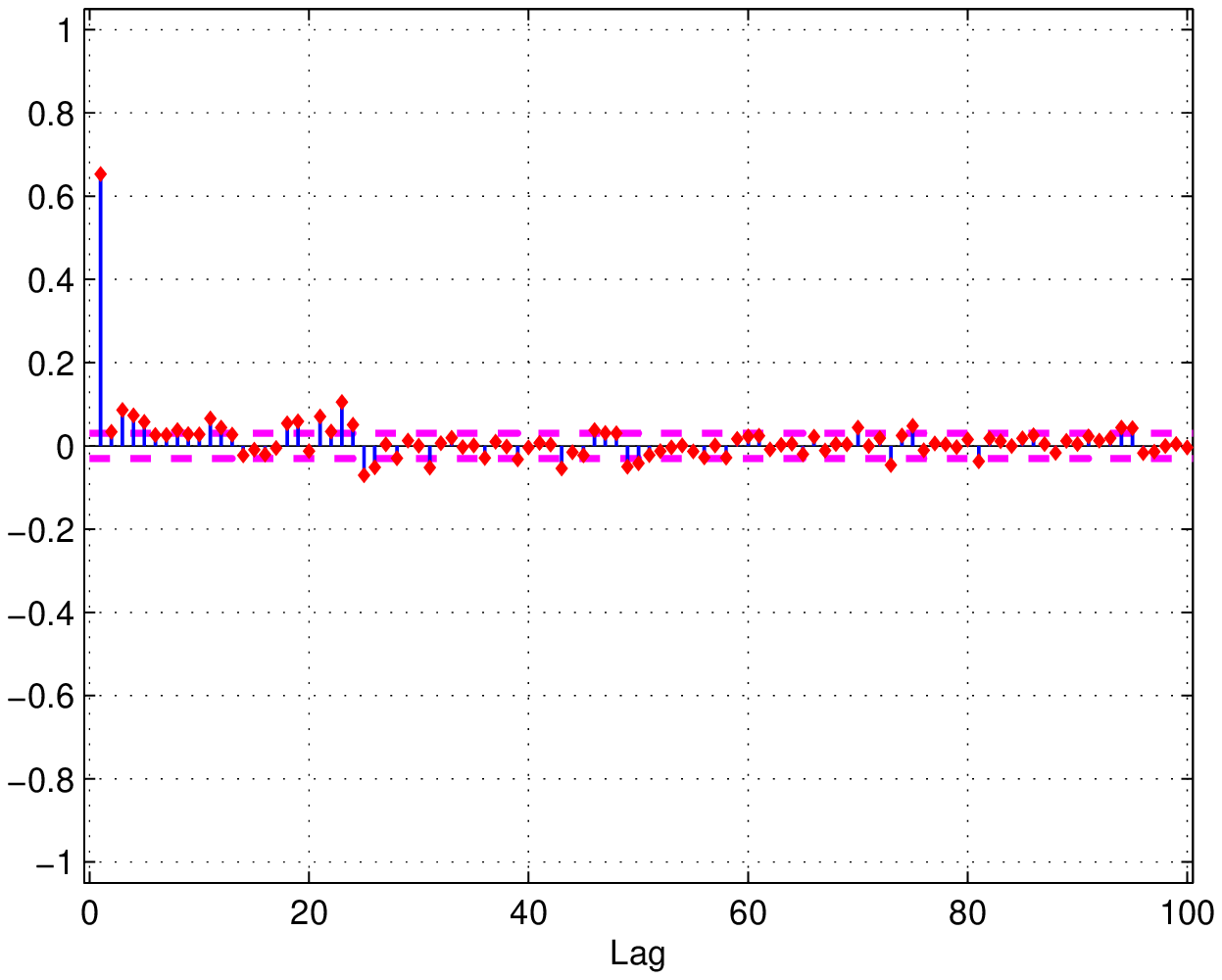}
\caption{ACF (left) and PACF (right) of the seasonally differenced residuals of period 24 $(\nabla_{\! 24} \wh{\veps}_{t})$.}
\label{Fig_ResDiffsACF}
\end{figure}

\medskip

\noindent Finally, Figure \ref{Fig_ResDiff2ACF} displays the sample ACF and PACF of $(\nabla \nabla_{\! 24} \wh{\veps}_{t})$. One can observe that the PACF tails off exponentially from lag 1 and that the ACF cuts off after lag 2, with small seasonal contributions. As a consequence, the series is likely to be generated by a SARIMA$(p,1,q) \times (0,1,Q)_{24}$ process with $p=1$, $q=2$ and $Q = 1$.
\begin{figure}[h!]
\includegraphics[width=7.2cm]{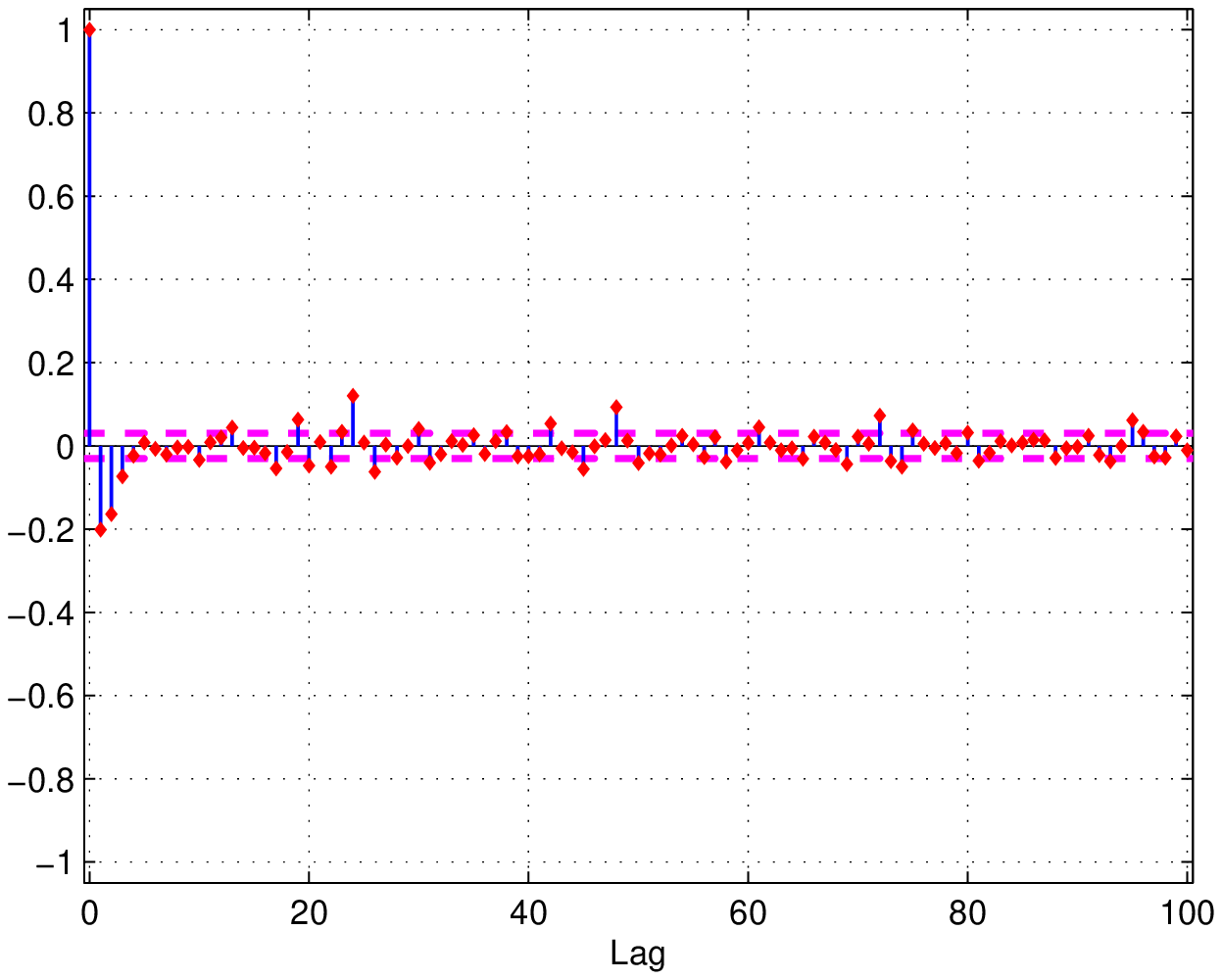} \includegraphics[width=7.2cm]{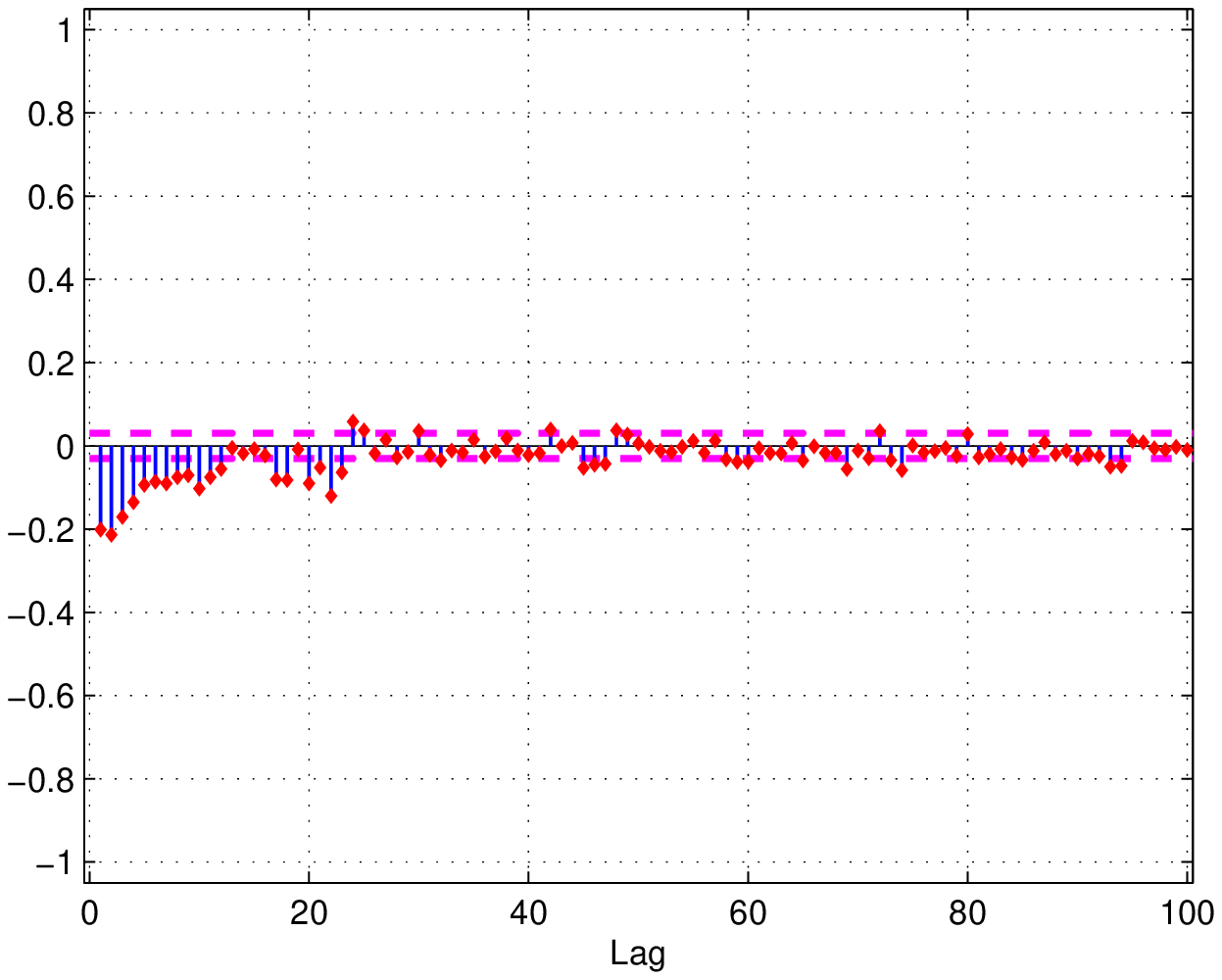}
\caption{ACF (left) and PACF (right) of the doubly differenced residuals of period 24 $(\nabla \nabla_{\! 24} \wh{\veps}_{t})$.}
\label{Fig_ResDiff2ACF}
\end{figure}

\medskip

\noindent{\bf Modelling.}  This identification methodology seems quite rough, and one shall make the parameters vary in their neighborhood to determine the optimal modelling. Table \ref{Tab_ModCrit} gives the bayesian criteria associated with a set of SARIMAX models fitted on 6 months of consumption. AIC, SBC, LL and WN respectively stand for \textit{Akaike Information Criterion}, \textit{Schwarz Bayesian Criterion}, \textit{Log-Likelihood} and \textit{White Noise}. Let us recall that
\begin{equation*}
\textnormal{AIC} = -2 \log \cL + 2k  \hspace{1cm} \textnormal{and} \hspace{1cm} \textnormal{SBC} = -2 \log \cL + k \log T
\end{equation*}
where $\cL$ is the model likelihood and $k$ is the number of parameters. In addition, VAR is the estimated variance of $(V_{t})$. The Ljung-Box \textit{portmanteau} test is used to evaluate the hypothesis of white noise on the fitted innovations, considering arbitrarily that $(V_{t})$ is a white noise if it has nonsignificant autocorrelations up to lag 3.

\begin{table}[h!]
\tiny
\begin{tabular}{|c|c|c|c|c|c|c|c|c||c|c|c|c|c|}
\cline{2-14}
\multicolumn{1}{c|}{} & $p$ & $d$ & $q$ & $r$ & $P$ & $D$ & $Q$ & $s$ & AIC & SBC & LL & VAR & WN \\
\hline
SARIMAX & 1 & 0 & 0 & 1 & 0 & 1 & 1 & 24 & -362.4 & -330.4 & 186.2 & 0.053 & \\
\hline
SARIMAX & 3 & 0 & 0 & 1 & 0 & 1 & 1 & 24 & -393.6 & -348.8 & 203.8 & 0.053 & \checkmark \\
\hline
SARIMAX & 5 & 0 & 0 & 1 & 0 & 1 & 1 & 24 & -424.7 & -367.2 & 221.4 & 0.053 & \checkmark \\
\hline
SARIMAX & 3 & 0 & 2 & 1 & 0 & 1 & 1 & 24 & -446.9 & -389.4 & 232.5 & 0.052 & \checkmark \\
\hline
SARIMAX & 3 & 0 & 2 & 2 & 0 & 1 & 1 & 24 & -457.2 & -393.3 & 238.6 & 0.052 & \checkmark \\
\hline
SARIMAX & 0 & 1 & 2 & 1 & 0 & 1 & 1 & 24 & -238.7 & -200.4 & 125.4 & 0.055 & \\
\hline
SARIMAX & 1 & 1 & 1 & 1 & 0 & 1 & 1 & 24 & -389.5 & -351.2 & 200.8 & 0.053 & \\
\hline
SARIMAX & 2 & 1 & 2 & 1 & 0 & 1 & 1 & 24 & -435.1 & -384.0 & 225.6 & 0.052 & \checkmark \\
\hline
\end{tabular}
\vspace{0.2cm} \\
\caption{Bayesian criteria associated with a set of SARIMAX models on 6 months of consumption.}
\normalsize
\label{Tab_ModCrit}
\end{table}
\vspace{-0.5cm}

\noindent Estimations come from an optimized mixing of conditional sum-of-squares and maximum likelihood \cite{BrockwellDavis96}, \cite{DurbinKoopman01}, \cite{GardnerHarveyPhillips80}, \cite{Harvey93}. In order to keep this section brief, only the most representative results are summarized in the table above even if more models have been evaluated. In conclusion, on the basis of the bayesian criteria, the most adequate modelling for the given load curve is a SARIMAX$(3,0,2,2) \times (0,1,1)_{24}$ whose explicit expression is as follows, for all $28 \leq t \leq T=4380$,
\begin{equation}
\label{ModEst}
\vspace{1ex}
\left\{
\begin{array}[c]{l}
Y_{t} = c_0 + c_1 U_{t} + c_2 U_{t-1} + \veps_{t}, \vspace{0.2cm}\\
\veps_{t} = \veps_{t-24} + a_1 ( \veps_{t-1} - \veps_{t-25} ) + a_2 ( \veps_{t-2} - \veps_{t-26} ) + a_3 ( \veps_{t-3} - \veps_{t-27} ) \\
\hspace{1cm} + ~ ( V_{t} - b_1 V_{t-1} - b_2 V_{t-2} ) - \beta_1 ( V_{t-24} - b_1 V_{t-25} - b_2 V_{t-26} ),
\end{array}
\right.
\end{equation}
in which the estimates at stage $T=4380$ are approximately given by
\begin{equation*}
\wh{c}_{0} = 7.9871, \hspace{0.3cm} \wh{c}_{1} = 0.0166, \hspace{0.3cm} \wh{c}_{2} = -0.0420, \hspace{0.3cm} \wh{a}_{1} = 0.4776, \hspace{0.3cm} \wh{a}_{2} = 0.9030,
\end{equation*}
\begin{equation*}
\wh{a}_{3} = -0.4305, \hspace{0.3cm} \wh{b}_{1} = 0.0801, \hspace{0.3cm} \wh{b}_{2} = -0.8524, \hspace{0.3cm} \wh{\beta}_{1} = -0.8125, \hspace{0.3cm} \wh{\sigma}^{\, 2} = 0.0522,
\end{equation*}
and the $t-$statistics of the time series coefficients, justifying their significance, by
\begin{equation*}
t_{\wh{a}_{1}} = 12.58, \hspace{0.3cm} t_{\wh{a}_{2}} = 26.13, \hspace{0.3cm} t_{\wh{a}_{3}} = 15.53, \hspace{0.3cm} t_{\wh{b}_{1}} = 2.50, \hspace{0.3cm} t_{\wh{b}_{2}} = 28.17, \hspace{0.3cm} t_{\wh{\beta}_{1}} = 75.42.
\end{equation*}

\medskip

\noindent Moreover, one can easily check that the estimation of the autoregressive polynomial $\wh{\cA}(z) = 1 - \wh{a}_{1} z - \wh{a}_{2} z^2 - \wh{a}_{3} z^3$ is causal, for all $z \in \dC$. The fitted values $(\wh{C}_{t})$ are obtained \textit{via} \eqref{Y}, that is, for all $28 \leq t \leq T$,
\begin{equation}
\label{ModEstC}
\wh{C}_{t} = \mathrm{e}^{\wh{Y}_{t}} - \mathrm{e}^{\mu}
\end{equation}
with $\mu=5$. On Figures \ref{Fig_ModLog} and \ref{Fig_Mod}, the fitted values $(\wh{Y}_{t})$ and $(\wh{C}_{t})$ from model \eqref{ModEst} and \eqref{ModEstC} are represented over the logarithmic load curve $(Y_{t})$ and the real load curve $(C_{t})$, respectively, with a zoom. The temperature $(U_{t})$ during the same period is represented on Figure \ref{Fig_Temp}. One can see that, except for some unpredictable local behaviors related to the individual nature of the curve, there is a pretty good adequation between modeled and real values.
\begin{figure}[h!]
\includegraphics[width=15cm]{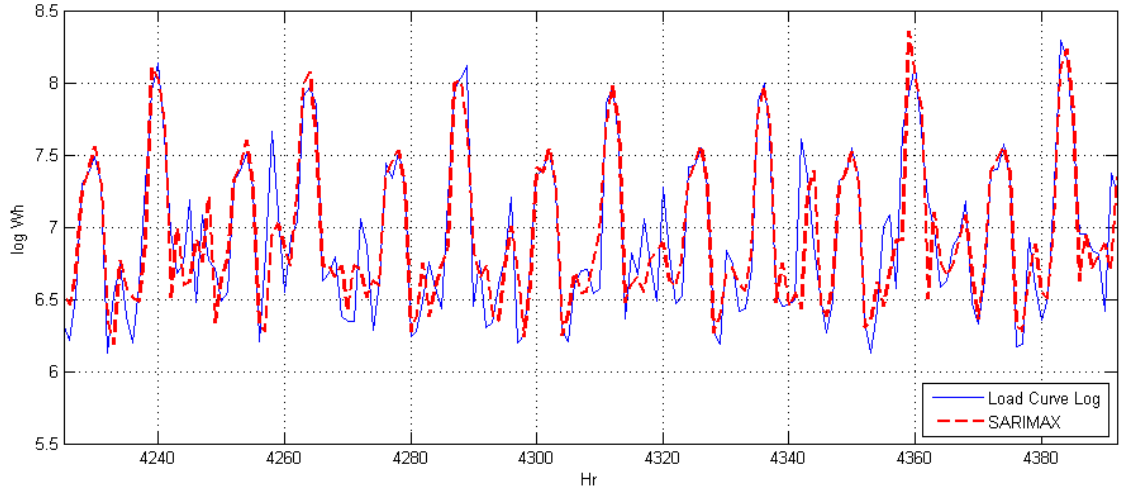}
\caption{SARIMAX$(3,0,2,2) \times (0,1,1)_{24}$ modelling on the logarithmic curve in red over the observed values in blue.}
\label{Fig_ModLog}
\end{figure}
\begin{figure}[h!]
\includegraphics[width=15cm]{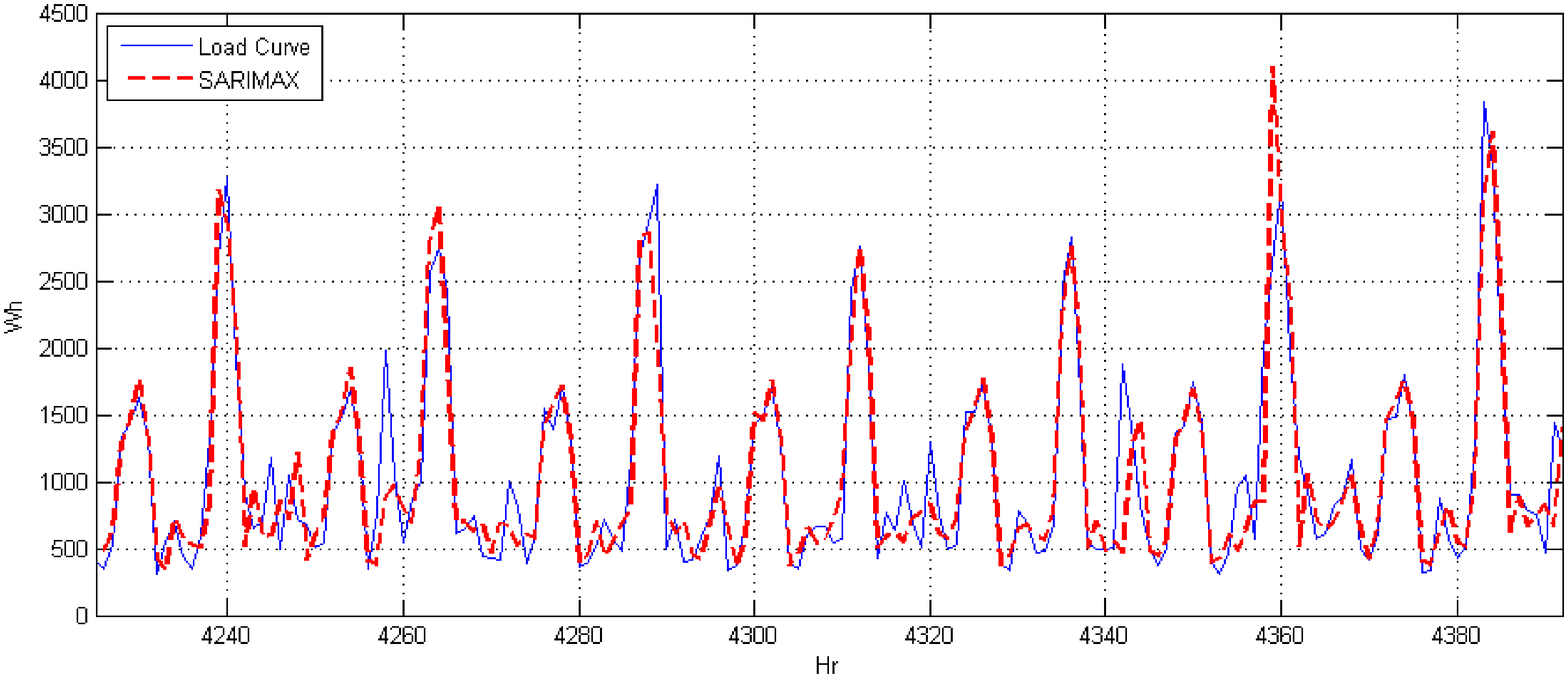}
\caption{SARIMAX$(3,0,2,2) \times (0,1,1)_{24}$ modelling on the real curve in red over the observed values in blue.}
\label{Fig_Mod}
\end{figure}
\begin{figure}[h!]
\includegraphics[width=15cm]{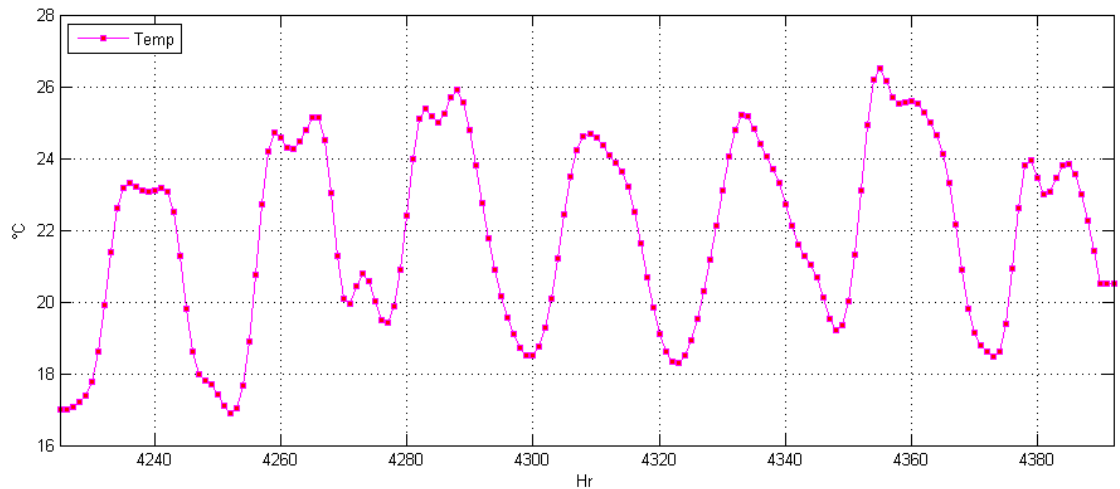}
\vspace{-0.8cm}
\caption{Temperature measured on the same period by the nearest weather station.}
\label{Fig_Temp}
\end{figure}

\medskip

\noindent{\bf Forecasting.} Our goal is now to propose an intraday forecasting methodology for the electrical consumption of individual customers. Let us start by introducing two criteria that will help us to select the most suitable forecasting model. Denote by $\wt{C}_{T+1}, \hdots, \wt{C}_{T+ N H}$ the values of $N$ consecutive predictions at horizon $H$ from time $T$. Then, the absolute criterion $C_{A}$ and the relative criterion $C_{R}$ are defined as follows,
\begin{equation*}
C_{A} = \frac{1}{N H} \sum_{k=1}^{N H} \left\vert \wt{C}_{T+k} - C_{T+k} \right\vert \hspace{0.5cm} \text{and} \hspace{0.5cm} C_{R} = \left(\sum_{k=1}^{N H} C_{T+k} \right)^{\hspace{-0.2cm} -1} \sum_{k=1}^{N H} \left\vert \wt{C}_{T+k} - C_{T+k} \right\vert.
\end{equation*}
Following the same lines as in the identification step, one has to make the parameters vary in their neighborhood to determine the most powerful forecasting model, considering the SARIMAX$(3,0,2,2) \times (0,1,1)_{24}$ modelling as a basis. A Kalman filtering finite-history prediction method \cite{DurbinKoopman01}, \cite{Harvey93}, \cite{HarveyMcKenzie82} is used to produce $(\wt{Y}_{T+k})$ from the modelling, for all $1 \leq k \leq N H$, and the forecasts $(\wt{C}_{T+k})$ are obtained by
\begin{equation}
\label{PredC}
\wt{C}_{T+k} = \mathrm{e}^{\wt{Y}_{T+k}} - \mathrm{e}^{\mu}
\end{equation}
with $\mu=5$. 

\begin{rem}
It is important to note that $\wt{C}_{T+1}, \hdots, \wt{C}_{T+ H}$ is not a sequence of predictions at horizon $H$ but a sequence in which only the last component is a prediction at horizon $H$. By misuse of language, one shall consider in the sequel that a sequence of predictions at horizon $H$ corresponds to $H$ successive predictions without additional meantime information. By extension, a sequence of $N$ predictions at horizon $H$ needs $N$ estimations of the parameters.
\end{rem}

\noindent Our experiments are based on $N=14$ days of daily forecasting, \textit{i.e.} $H=24$, the coefficients are evaluated on 3 months of data, that is $T = 2190$, and the numerical results are summarized in the Table \ref{Tab_PredCrit} below.

\medskip

\begin{table}[h!]
\tiny
\begin{tabular}{|c|c|c|c|c|c|c|c|c||c|c|}
\cline{2-11}
\multicolumn{1}{c|}{} & $p$ & $d$ & $q$ & $r$ & $P$ & $D$ & $Q$ & $s$ & $C_{A}$ & $C_{R}$ \\
\hline
SARIMAX & 1 & 0 & 0 & 1 & 0 & 1 & 1 & 24 & 241.0 & 0.2279 \\
\hline
SARIMAX & 1 & 0 & 1 & 2 & 0 & 1 & 1 & 24 & 242.2 & 0.2290 \\
\hline
SARIMAX & 3 & 0 & 0 & 1 & 0 & 1 & 1 & 24 & 245.1 & 0.2318 \\
\hline
SARIMAX & 3 & 0 & 2 & 1 & 0 & 1 & 1 & 24 & 251.6 & 0.2380 \\
\hline
SARIMAX & 3 & 0 & 2 & 2 & 0 & 1 & 1 & 24 & 250.3 & 0.2368 \\
\hline
SARIMAX & 1 & 1 & 1 & 1 & 0 & 1 & 1 & 24 & 253.8 & 0.2400 \\
\hline
SARIMAX & 2 & 1 & 2 & 1 & 0 & 1 & 1 & 24 & 254.1 & 0.2403 \\
\hline
\end{tabular}
\vspace{0.2cm} \\
\caption{Prediction criteria associated with a set of daily SARIMAX forecasts on 3 months of consumption.}
\normalsize
\label{Tab_PredCrit}
\end{table}

\noindent The parsimony in the time series analysis is a central issue in forecasting applications, and it is not surprising that the models minimizing $C_{A}$ and $C_{R}$ are not the same as those minimizing the bayesian criteria, and tend to reject uncertainty coming from overparametrization. Moreover, by selecting an optimal sliding window in the modelling, one is able to slightly improve our results. For example, the SARIMAX$(1,0,0,2) \times (0,1,1)_{24}$ model provides $C_{A}=231.0$ and $C_{R}=0.2185$ in the particular case where $M=1$ month, that is $T=730$. On Figure \ref{Fig_PredMR}, we investigate the influence of the size of the sliding window $M$ together with the one of the exogenous regression dimension $r$ on the relative criterion $C_{R}$ for the latter modelling and the same experiment. This enables us to select the most powerful forecasting model for this particular curve.
\begin{figure}[h!]
\includegraphics[width=14cm]{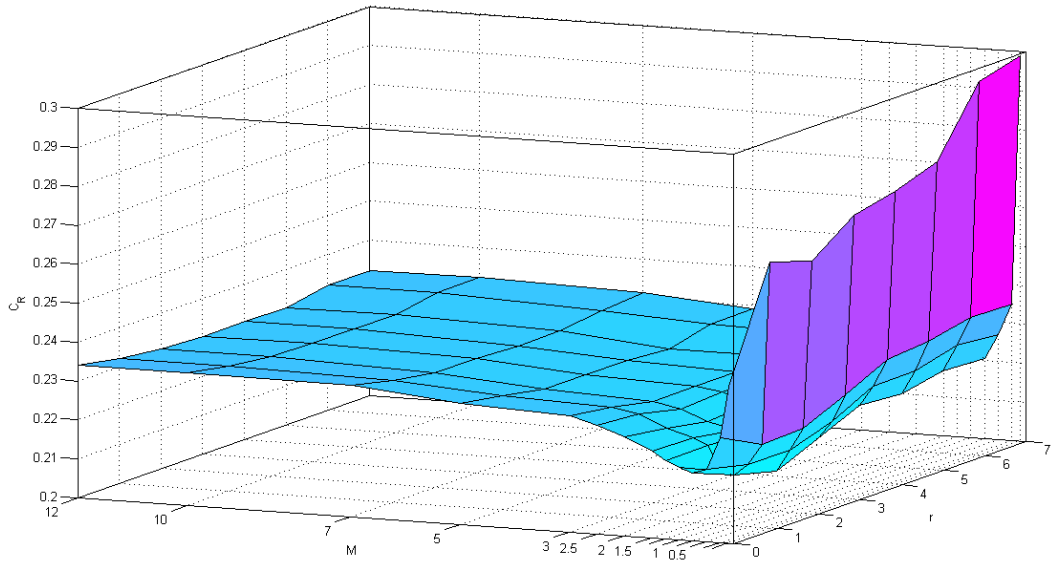}
\vspace{-0.2cm} \\
\caption{Influence of $r$ and $M$ on $C_{R}$ calculated from 14 daily forecasts from the SARIMAX$(1,0,0,r) \times (0,1,1)_{24}$ modelling.}
\label{Fig_PredMR}
\end{figure}

\medskip

\noindent Accordingly, one shall consider the SARIMAX$(1,0,0,2) \times (0,1,1)_{24}$ modelling with $M=0.75$ month, even if one can see that $r$ is not playing a substantial role as soon as it is greater than 2 for a reason of strong correlation of the exogenous phenomenon, already mentioned above. One can also observe that prediction results can be improved when the parameters are evaluated on rather small amounts of data. This can once again be explained by the nature of the curve and the underlying human behavior whose consumption is highly influenced by local circumstances such as weather, holiday period, etc. Whereas too few data are  not sufficient to take into account the seasonality of period 24 and properly estimate the very significant $\beta_1$ parameter, conversely, results tend to stabilize when $M$ increases disproportionately. The explicit expression of the predictive model is as follows, for all $26 \leq t \leq T=548$,
\begin{equation}
\label{ModPred}
\vspace{1ex}
\left\{
\begin{array}[c]{l}
Y_{t} = c_0 + c_1 U_{t} + c_2 U_{t-1} + \veps_{t}, \vspace{0.2cm}\\
\veps_{t} = \veps_{t-24} + a_1 ( \veps_{t-1} - \veps_{t-25} ) + V_{t} - \beta_1 V_{t-24},
\end{array}
\right.
\end{equation}
in which the estimates at stage $T=548$ are approximately given by
\begin{equation*}
\wh{c}_{0} = 7.2494, \hspace{0.15cm} \wh{c}_{1} = 0.0497, \hspace{0.15cm} \wh{c}_{2} = -0.0629, \hspace{0.15cm} \wh{a}_{1} = 0.3540, \hspace{0.15cm} \wh{\beta}_{1} = -0.7086, \hspace{0.15cm} \wh{\sigma}^{\, 2} = 0.0708,
\end{equation*}
and the $t-$statistics of the time series coefficients, justifying their significance, by
\begin{equation*}
t_{\wh{a}_{1}} = 8.66, \hspace{0.3cm} t_{\wh{\beta}_{1}} = 21.02.
\end{equation*}

\medskip

\noindent The estimation of the autoregressive polynomial $\wh{\cA}(z) = 1 - \wh{a}_{1} z$ is actually causal, for all $z \in \dC$. On Figures \ref{Fig_PredLog} and \ref{Fig_Pred}, we display an example of 7 daily predictions from the latter model and a sliding window of 0.75 month of consumption, for the logarithmic curve $(\wt{Y}_{t})$ as well as for the load curve $(\wt{C}_{t})$. It also contains the 95\% and 90\% prediction confidence intervals, rather large owing to the horizon of prediction.
\begin{figure}[h!]
\includegraphics[width=15cm]{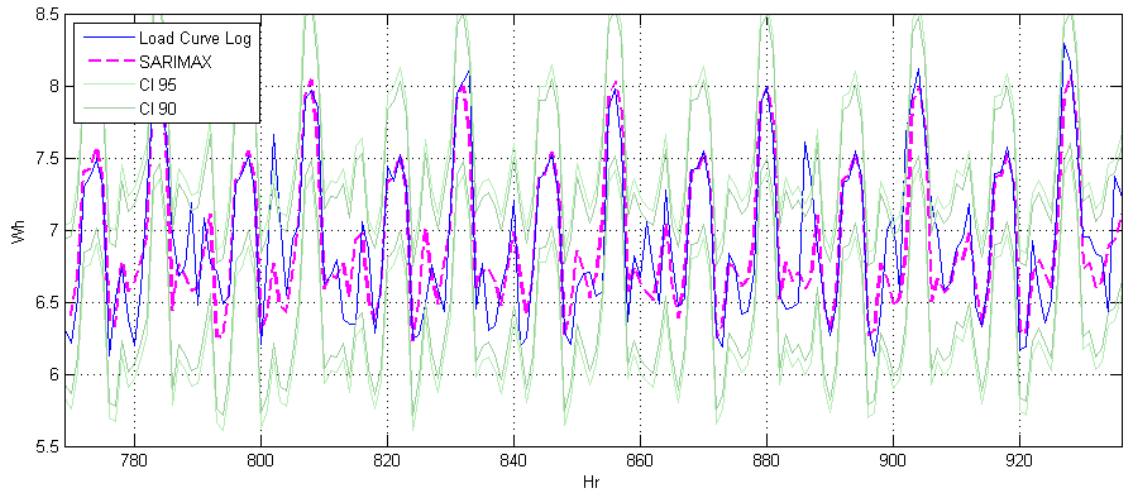}
\caption{SARIMAX$(1,0,0,2) \times (0,1,1)_{24}$ daily predictions on the logarithmic curve in magenta over the observed values in blue.}
\label{Fig_PredLog}
\end{figure}

\begin{figure}[h!]
\includegraphics[width=15cm]{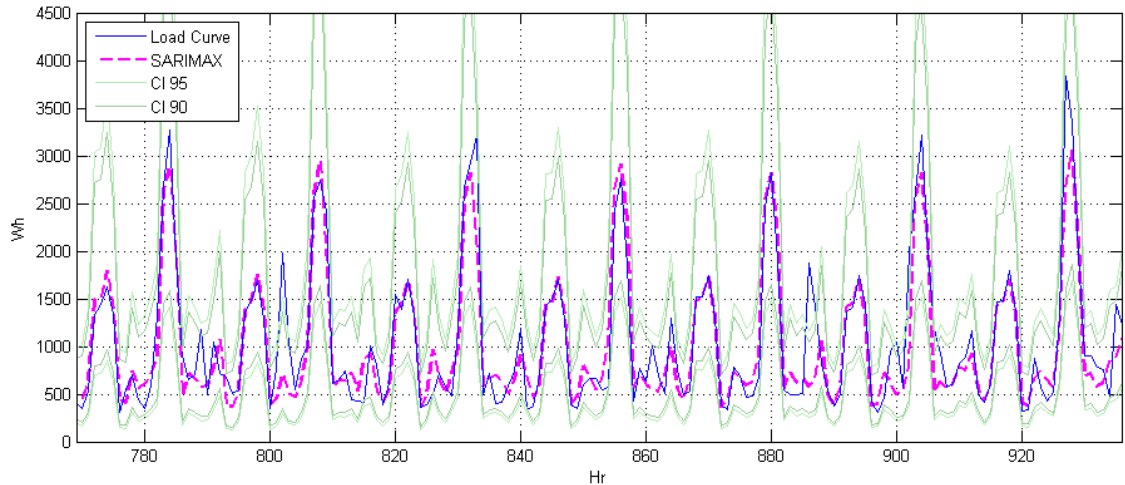}
\caption{SARIMAX$(1,0,0,2) \times (0,1,1)_{24}$ daily predictions on the real curve in magenta over the observed values in blue.}
\label{Fig_Pred}
\vspace{1.5cm}
\end{figure}

\medskip

\noindent Note that $\wt{Y}_{T+1}, \hdots, \wt{Y}_{T+ N H}$ only differs of 3\% from $Y_{T+1}, \hdots, Y_{T+ N H}$ when we consider the whole 14 daily predictions, that is $NH=336$. Once again, one can notice that the coupled dynamic model provides very interesting results of prediction as soon as it is correctly specified, in spite of the noise on the load curve due to its individual nature.


\section{CONCLUSION}


To conclude, we would like to draw the significance of the exogenous covariates to the reader's attention. Indeed, let us notice that in some cases, the empirical study suggests to select $r=0$, meaning no temperature influence despite the manifest linear relation between the latter and the consumption. The authors interpret this observation by the fact that seasonality and local circumstances totally prevail on the effect of temperature and that all information has already been recovered by the deep study of the signal, resulting in very few significant coefficients and equivalent forecasts for $r=0, 1, \hdots$ In addition, one should not overlook the possible irrelevance of the temperature measured by the weather station related to a customer without other criterion than geolocation, especially when altitude is concerned, coastal residence, cloud covering, or more generally when substantial differences may be observed at the same time between the weather station and the customer's home. Also, we should not forget that the exogenous inputs assumed to be known during the prediction period of the time series are nothing else than predictions themselves, with all attendant uncertainty. In addition to the required seasonality, the relevance of the exogenous measures is a central issue for this approach to be applied to a load curve with profit. Nevertheless, and despite many irregularities due to the individual nature of the curves, this study shows that some very interesting results of daily forecasts may be obtained under certain conditions already described, and above all a careful study of each curve. Finally, this intraday forecasting approach has been conducted
on a whole set of individual customers from EDF, leading to
the same satisfactory conclusions.

\newpage

\nocite{*}

\bibliographystyle{acm}
\bibliography{SARIMAX0312}

\vspace{10pt}

\end{document}